\documentclass[11pt]{article}
\usepackage{amsmath,amsfonts,mathrsfs,amssymb}
\usepackage{theorem}
\newcommand{\N}{{\mathbb N}}
\newcommand{\F}{{\mathbb F}}

\newcommand{\Z}{{\mathbb Z}}
\newcommand{\cS}{{\mathcal S}}
\newcommand{\cP}{{\mathcal P}}

\newcommand{\cD}{{\mathcal D}}
\newcommand{\cM}{{\mathcal M}}
\newcommand{\cC}{{\mathcal C}}
\newcommand{\Znn}{\mbox{$\Z_n^{n-1}$}}

\newcommand{\rank}{\mbox{\rm rank}\,}

\newcommand{\Supp}{\mbox{\rm Supp}\,}

\newcommand{\AutF}{\mbox{${\rm Aut}_{\mathbb F}$}}
\newcommand{\cMbasic}{\mbox{$\cM_{\text{basic}}$}}
\newcommand{\im}{\mbox{\rm im}\,}
\renewcommand{\mod}{\,\text{mod}\,}
\newcommand{\sgn}{\text{sgn}}

\newcommand{\dist}{\mbox{\rm dist}}
\newcommand{\wt}{\mbox{\rm wt}}

\newcommand{\one}{\mbox{$\bf 1$}}

\newcommand{\ideal}[1]{\mbox{$\langle{#1}\rangle$}}

\newcommand{\Azs}{\mbox{$A[z;\sigma]$}}

\newcommand{\p}{\mbox{$\mathfrak{p}$}}

\newcommand{\lideal}[1]{\mbox{$^{^{\bullet\!\!}}\langle{\, #1\, }\rangle$}}
\newcommand{\rideal}[1]{\mbox{$\langle{\, #1\, }\rangle^{^{\!\bullet}}$}}

\newcounter{abc}
\newcounter{def}
\newenvironment{romanlist}{\begin{list}{(\roman{abc})\hfill}{\usecounter{abc}
     \topsep-1.3ex \labelwidth.8cm \leftmargin.8cm \labelsep0cm
     \rightmargin0cm \parsep0ex \itemsep.6ex
     \partopsep1.6ex}}{\end{list}}
\newenvironment{alphalist}{\begin{list}{(\alph{abc})\hfill}{\usecounter{abc}
     \topsep-1.3ex \labelwidth.7cm \leftmargin.7cm \labelsep0cm
     \rightmargin0cm \parsep0ex \itemsep.6ex
     \partopsep1.6ex}}{\end{list}}
\newenvironment{arabiclist}{\begin{list}{(\arabic{abc})\hfill}{\usecounter{abc}
     \topsep-1.3ex \labelwidth.7cm \leftmargin.7cm \labelsep0cm
     \rightmargin0cm \parsep0ex \itemsep.6ex
     \partopsep1.6ex}}{\end{list}}

\theoremheaderfont{\bf} 
\theorembodyfont{\sl}
\newtheorem{theo}{Theorem}[section]
{\theorembodyfont{\rm} \newtheorem{defi}[theo]{Definition}}
{\theorembodyfont{\rm} \newtheorem{exa}[theo]{Example}}
{\theorembodyfont{\rm} \newtheorem{rem}[theo]{Remark}}
\newtheorem{prop}[theo]{Proposition}
\newtheorem{cor}[theo]{Corollary}
\newtheorem{lemma}[theo]{Lemma}
\newtheorem{prob}[theo]{Problem}
{\theorembodyfont{\rm}}
{\theorembodyfont{\rm}}
\newenvironment{proof}{{\sc Proof:}}{\mbox{}\hfill$\Box$\par}
\newcommand{\eqnref}[1]{~\mbox{$(${\rm \ref{#1}}$)$}}

\newcommand{\set}[1]{\left\{#1\right\}}

\begin{document}
\title{A Matrix Ring Description for Cyclic Convolutional Codes}
\date\today
\author{Heide Gluesing-Luerssen\footnote{University of Kentucky,
           Department of Mathematics, 715 Patterson Office Tower, Lexington,
           KY 40506- 0027, USA; heidegl@ms.uky.edu},
        Fai-Lung Tsang\footnote{University of Groningen,
           Department of Mathematics, P.~O.~Box 800, 9700 AV Groningen,
           The Netherlands; F.L.Tsang@math.rug.nl}
       }


\maketitle
{\bf Abstract:}
In this paper, we study convolutional codes with a specific cyclic structure.
By definition, these codes are left ideals in a certain skew polynomial ring.
Using that the skew polynomial ring is isomorphic to a matrix ring we can
describe the algebraic parameters of the codes in a more accessible way.
We show that the existence of such codes with given algebraic
parameters can be reduced to the solvability of a modified rook problem.
It is our strong belief that the rook problem is always solvable, and we present
solutions in particular cases.



\section{Introduction}\label{S-Intro}
\setcounter{equation}{0}
Convolutional codes (CC's, for short) form an important class of error-correcting codes
in engineering practice.
The mathematical theory of these codes has been set off by the seminal papers of
Forney~\cite{Fo70} and Massey et~al.~\cite{MaSa67}, and the progress ever since
is reflected by, for instance, the books~\cite{JoZi99,Pi88b} and the article~\cite{McE98}.
Several directions have been pursued.
In the 1970s, a lot of effort has been made to construct powerful CC's with
the aid of good block codes, see~\cite{MCJ73,Ju75}.
This idea has been resumed in the papers~\cite{SGR01,GRS06}.
Furthermore, the methods of linear systems theory have been utilized
in order to gain a deeper mathematical understanding of CC's.
We refer to the papers \cite{MaSa67,RSY96,RoYo99,HRS05} for further details and
constructions.
A third direction in the theory of CC's developed when codes with some
additional algebraic structure were studied.
Besides the recently introduced classes of Goppa convolutional codes~\cite{PPS04,MPCS06}
and group convolutional codes~\cite{EGPS07}, the main class of such codes are cyclic
convolutional codes.

Cyclic structure for CC's has been investigated for the first time in the
papers~\cite{Pi76,Ro79}.
One of the crucial observations revealed that CC's that are invariant
under the ordinary cyclic shift have degree zero, that is, they are cyclic
block codes.
This insight has led to a more complex notion of cyclicity for CC's which can be
summarized as follows.
Cyclic convolutional codes (CCC's, for short) are direct summands of $\F[z]^n$ that are
at the same time left ideals in a skew polynomial ring $\Azs$, where
$A=\F[x]/\ideal{x^n-1}$ and~$\sigma$ is an $\F$-automorphism of~$A$.
During the last couple of years a detailed algebraic theory of CCC's has been developed
in \cite{GS04,GL06}.
Among other things it has been shown that CCC's are principal left ideals in $\Azs$ and,
using a type of Gr\"obner basis theory, one can compute reduced generator
polynomials from which all algebraic parameters of the code can easily be read off.
The details will be given later on in Theorem~\ref{T-ideals}.
Classes of CCC's with good error-correcting properties have been
presented in \cite{GL06p,GS06p}.

In this paper we want to continue the investigation of CCC's.
We restrict ourselves to a particular class of automorphisms.
In that case the skew polynomial ring $\Azs$ turns out to be isomorphic to a matrix
ring over a commutative polynomial ring.
This allows us to easily construct generator polynomials of CCC's with prescribed
algebraic parameters.
Moreover, we discuss the existence of CCC's with any given algebraic
parameters and will show that it reduces to a particular combinatorial problem
resembling the classical rook problem.
We can solve particular instances of that problem, but unfortunately not the general
case.
However, we strongly believe that the combinatorial problem is solvable for all
possible cases.

The outline of the paper is as follows.
In the next section we review important notions from convolutional coding theory and
introduce the algebraic framework for CCC's.
In Section~\ref{S-M} we concentrate on a particular class of CCC's and
establish an isomorphism between the associated skew polynomial ring and a
certain matrix ring.
We translate the main notions needed for the theory of CCC's into the matrix setting.
In Section~\ref{S-CCC} we construct particular CCC's and discuss the existence of
CCC's with given Forney indices.
The existence of such codes reduces to a combinatorial problem followed by
a problem of constructing polynomial matrices with certain degree properties.
After presenting a proof of the matrix problem we discuss and solve particular instances
of the combinatorial problem in Section~\ref{S-rook}.
We close the paper with a short section illustrating how to generalize the results to
codes that are cyclic with respect to a general automorphism.

\section{Preliminaries on Cyclic Convolutional Codes}
\setcounter{equation}{0}
Throughout this paper let $\F$ be a finite field with~$q$ elements.
A {\sl convolutional code of length\/} $n$ and dimension~$k$ is a submodule~$\cC$
of $\F[z]^n$ having the form
\[
    \cC=\im G:=\{uG\,\big|\, u\in\F[z]^k\},
\]
where~$G$ is a {\sl basic\/} matrix in $\F[z]^{k\times n}$, i.~e.
$\rank G(\lambda)=k$ for all $\lambda\in\overline{\F}$
(with $\overline{\F}$ being an algebraic closure of~$\F$).
We call such a matrix $G$ an {\sl encoder}, and the number
$\deg(\cC):=\deg(G):=\max\{\deg(M)\mid M\text{ is a $k$-minor of }G\}$
is said to be the {\sl degree\/} of the encoder~$G$ or of the code~$\cC$.
Recall that the requirement of~$G$ being basic (rather than just having full row
rank over $\F[z]$) is equivalent to~$\cC$ being a direct summand of the module
$\F[z]^n$.
Obviously, two basic matrices $G,\,G'\in\F[z]^{k\times n}$ satisfy
$\im G=\im G'$ if and only if $G'=UG$ for some
$U\in GL_k(\F[z]):=\{V\in\F[z]^{k\times k}\mid \det(V)\in\F\backslash\{0\}\}$.
It is well known~\cite[p.~495]{Fo75} that each CC admits a minimal encoder.
Here a matrix $G\in\F[z]^{k\times n}$ is said to be
{\sl minimal\/} if the sum of its row degrees equals $\deg(G)$, where the degree
of a polynomial row vector is defined as the maximal degree of its entries.
For details see, e.~g., \cite[Main~Thm.]{Fo75} or \cite[Thm.~A.2]{McE98}.
Such a minimal encoder is in general not unique; the row degrees of a minimal
encoder, however, are, up to ordering, uniquely determined by the code and are
called the {\sl Forney indices\/} of the code or of the encoder.
It follows that a CC has a constant encoder matrix if and only
if the degree is zero.
In that case the code is, in a natural way, a block code.

Beyond these purely algebraic concepts the most important notion in
error-control coding is certainly the weight.
For a polynomial vector $v=\sum_{j=0}^N v^{(j)}z^j\in\F[z]^n$, where $v^{(j)}\in\F^n$,
one defines its weight as $\wt(v)=\sum_{j=0}^N \wt(v^{(j)})$, with
$\wt(v^{(j)})$ denoting the (Hamming) weight of the vector $v^{(j)}\in\F^n$.
Just like for block codes the {\sl distance\/} of a code~$\cC$ is defined as
$\dist(\cC)=\min\{\wt(v)\mid v\in\cC,\,v\not=0\}$.

Let us now turn to the notion of cyclicity (for details we refer to \cite{GS04,GL06}).
We will restrict ourselves to codes where the length~$n$ is coprime with the field
size~$q$.
From the theory of cyclic block codes recall the standard identification
\begin{equation}\label{eq:pF}
   \p:\; \F^n\longrightarrow A:=\F[x]/{\ideal{x^n-1}},\quad
   (v_0,\ldots,v_{n-1})\longmapsto\sum_{i=0}^{n-1}v_ix^i
\end{equation}
of~$\F^n$ with the ring of polynomials modulo $x^n-1$.
Extending this map coefficientwise we can identify the polynomial module
$\F[z]^n=\{\sum_{\nu=0}^Nz^{\nu}v_{\nu}\mid N\in\N_0,\,v_{\nu}\in\F^n\}$
with the polynomial ring
\[
   A[z]:=\Big\{\sum_{\nu=0}^Nz^{\nu}a_{\nu}\,\Big|\, N\in\N_0,\,a_{\nu}\in A\Big\}.
\]
Following the theory of cyclic block codes one would like to declare a code
$\cC\subseteq\F[z]^n$ cyclic if it is invariant under the cyclic shift acting on
$\F[z]^n$, or, equivalently, if its image in~$A[z]$ is an ideal.
However, it has been shown in various versions in~\cite[Thm.~3.12]{Pi76},
\cite[Thm.~6]{Ro79}, and \cite[Prop.~2.7]{GS04} that every CC with
this property has degree zero.
In other words, this notion does not lead to any codes other than cyclic block
codes.
Due to this result a more general notion of cyclicity has been introduced and
discussed in the papers mentioned above.
In order to present this notion we need the group $\AutF(A)$ of all
$\F$-automorphisms on~$A$.
It is clear that each automorphism $\sigma\in\AutF(A)$ is uniquely
determined by the single value $\sigma(x)\in A$, but not every choice for
$\sigma(x)$ determines an automorphism on~$A$.
Indeed, since~$x$ generates the $\F$-algebra~$A$ the same has to be true
for~$\sigma(x)$, and we obtain for $a\in A$
\begin{equation}\label{eq:sigma(x)}
  \left.\begin{array}{l} \sigma(x)=a\text{ determines an}\\
                    \text{$\F$-automorphism on }A\end{array}\right\}
  \Longleftrightarrow
  \left\{\begin{array}{l} 1, a,\ldots, a^{n-1}\text{ are linearly independent}\\
                    \text{over~$\F$ and }a^n=1.\end{array}\right.
\end{equation}
Fixing an arbitrary automorphism $\sigma\in\AutF(A)$ we define a new multiplication
on the $\F[z]$-module $A[z]$ via
\begin{equation}\label{eq:az}
    az=z\sigma(a) \text{ for all }a\in A.
\end{equation}
Along with associativity, distributivity, and the usual multiplication inside~$A$
this turns $A[z]$ into a skew polynomial ring which we will denote by $\Azs$.
Notice that $\Azs$ is non-commutative unless~$\sigma$ is the identity.
Moreover, the map~$\p$ from\eqnref{eq:pF} can be extended to
\begin{equation}\label{eq:p}
    \p:\F[z]^n\longrightarrow\Azs,\quad
    \sum_{\nu=0}^Nz^{\nu}v_{\nu}\longmapsto\sum_{\nu=0}^Nz^{\nu}\p(v_{\nu}),
\end{equation}
this way yielding an isomorphism of left $\F[z]$-modules.
Now we declare a submodule $\cC\subseteq\F[z]^n$ to be $\sigma$-{\em cyclic\/} if
$\p(\cC)$ is a left ideal in $\Azs$.
It is straightforward to see that the latter is equivalent to the
$\F[z]$-submodule $\p(\cC)$ of $\Azs$
being closed under left multiplication by~$x$.
Combining this with the definition of CC's we arrive
at the following notion.
\begin{defi}\label{D-CCC}
A submodule $\cC\subseteq\F[z]^n$ is called a $\sigma$-cyclic convolutional code
($\sigma$-CCC, for short) if $\cC$ is a direct summand of~$\F[z]^n$ and $\p(\cC)$
is a left ideal in $\Azs$.
\end{defi}
In the papers \cite{Pi76,GS04,GL06,GL06p,GS06p} the algebraic properties of these
codes have been investigated in detail (the main results are summarized in Theorem~\ref{T-ideals} below)
and plenty of CCC's,
all optimal with respect to their free distance, have been presented.
In~\cite{GL06p} a class of one-dimensional CCC's has been
constructed all of which members are MDS codes, that is, they have the best
distance among all one-dimensional codes with the same length, same degree and
over any finite field.
This result has been generalized to a class of Reed-Solomon convolutional codes
in~\cite{GS06p}.
In~\cite{EGPS07} the concept of cyclicity has been generalized to group
convolutional codes.

Let us close this section with some basic notation.
In any ring we will denote left, respectively right, ideals by
$\lideal{\;\cdot\;}$, respectively $\rideal{\;\cdot\;}$.
As usual, ideals in commutative rings will be written as $\ideal{\;\cdot\;}$.
The group of units in a ring~$R$ will be denoted by~$R^\times$.

\section{The Matrix Ring $\mathcal{M}$}
\label{S-M}
\setcounter{equation}{0}
In this paper we will restrict ourselves to a more specific class of cyclic codes
than those being introduced in the previous section.
Indeed, we will consider the following situation.
Let~$\F$ be any field with~$q$ elements and let~$n\in\N_{\geq2}$ be a divisor of
$q-1$.
Then, using the Chinese Remainder Theorem, the quotient ring $A':=\F[x]/\ideal{x^n-1}$
is isomorphic to the direct product
\begin{equation}\label{eq:A}
    A:=\underbrace{\F \times \cdots \times \F}_{n \text{ copies}}.
\end{equation}
The elements of~$A$ will be denoted as $a=[a_1,\dots,a_n]$, where
$a_i \in \F$ for $i=1,\ldots,n$.
Observe that~$\F$ naturally embeds into~$A$ via $f\longmapsto[f,\ldots,f]$,
thus~$A$ is an $\F$-algebra.
This algebra structure is isomorphic to the natural $\F$-algebra structure
of~$A'$.

The standard $\F$-basis of $A$ is given by $\set{e_1, \dots , e_n}$, where
\begin{equation}\label{eq:e}
    e_i = [0, \dots, 1, \dots, 0]
    \text{ with the $1$ appearing in the $i$-th position}.
\end{equation}
Obviously, these are just the primitive idempotents of the ring~$A$, and
$A=\oplus_{i=1}^n \ideal{e_i}$ as a direct sum of ideals.
It is clear that the automorphisms of~$A$ are in one-one
correspondence with the permutations of the primitive idempotents.
As a consequence, $|\AutF(A)|=n!$, see also \cite[Cor.~3.2]{GS04}.
In this paper, we will consider only those automorphisms for which the permutation
is a cycle of length~$n$.
It is not hard to see that there are exactly $(n-1)!$ automorphisms of that kind.
Allowing a suitable permutation of the~$n$ copies of~$\F$, we may restrict
ourselves to the automorphism
\begin{equation}\label{eq:sigma}
   \sigma:A\longrightarrow A,\qquad
   \sigma([a_1, \dots, a_n])=[a_n, a_1,\dots,a_{n-1}].
\end{equation}
Then $\sigma (e_i) = e_{i+1}$ for $1 \le i \le n-1$ and $\sigma(e_n)=e_1$, and
we have the identities
\begin{equation}\label{eq:sigmae}
  \sigma^j(e_i)=e_{(i+j-1\mod n)+1}\text{ for all }
  i=1,\ldots,n\text{ and }j\in\N_0.
\end{equation}
As in the previous section, the automorphism~$\sigma$ gives rise to a skew
polynomial ring $(\Azs,+,\,\cdot\,)$ where, again, the set
$\{\sum_{\nu=0}^N z^\nu a_\nu\mid N\in\N_0,\,a_\nu\in A\}$ is equipped with the
usual coordinatewise addition, and where multiplication is defined by the rule
$az=z\sigma(a)$ for all $a\in A$, see\eqnref{eq:az}.
As already indicated by the notation above, coefficients of the polynomials in
$\Azs$ are always meant to be the right hand side coefficients.
Again, the left $\F[z]$-module $\Azs$ is isomorphic to $\F[z]^n$
via the map~$\p$ from\eqnref{eq:p}, and where we use an isomorphism of~$A$ with the
quotient ring~$A'$.
Thus, $\Azs$ gives us the framework for the class of CCC's where the length~$n$ divides
$q-1$ and where the underlying automorphism induces a cycle of length~$n$ on the
primitive idempotents.

\begin{exa}\label{E-autom}
Let $\F=\F_4=\{0,1,\alpha,\alpha^2\}$, where $\alpha^2=\alpha+1$, and let $n=3$.
Then $A=\F\times\F\times\F\cong\F[x]/\ideal{x^3-1}=A'$ and $|\AutF(A)|=6$.
The automorphisms~$\sigma\in\AutF(A')$ are completely determined by their
value~$\sigma(x)$ assigned to~$x$, see also\eqnref{eq:sigma(x)}.
These values are given by
$x, x^2,\alpha x, \alpha^2 x, \alpha x^2, \alpha^2 x^2$.
The two automorphisms inducing cycles of length~$3$ are given by
$\sigma_1(x)=\alpha x$ and $\sigma_2(x)=\alpha^2 x$.
Indeed, using the isomorphism
$\phi:A'\longrightarrow A,\ f\longmapsto
   [f(1), f(\alpha), f(\alpha^2)]$
the primitive idempotents of~$A$ are
$e_1=[1,0,0]=\phi(x^2+x+1),\,e_2=[0,1,0]=\phi(\alpha x^2+\alpha^2 x+1)$, and
$e_3=[0,0,1]=\phi(\alpha^2 x^2+\alpha x+1)$.
One easily verifies that
$\sigma_2(e_1)=e_2,\,\sigma_2(e_2)=e_3,\, \sigma_2(e_3)=e_1$, thus~$\sigma_2$
satisfies\eqnref{eq:sigma} (where, of course, we identify~$\sigma_2\in\AutF(A')$
with $\phi\sigma_2\phi^{-1}\in\AutF(A)$).
Likewise, using the identification
$A'\longrightarrow A,\ f\longmapsto [f(\alpha^2), f(\alpha), f(1)]$
the automorphism~$\sigma_1$ satisfies\eqnref{eq:sigma}.
Notice that $\sigma_2=\sigma_1^{-1}$.
\end{exa}

In the rest of this section we will show how the skew polynomial ring $\Azs$
can be described as a certain matrix ring, and we will translate various
properties into the matrix setting.
This will lead us to a new way of constructing CCC's with
given algebraic parameters.
Let us consider the ring $\F[t]^{n\times n}$ of polynomial matrices in an
indeterminate~$t$ and define the subset
\[
   \cM:=\big\{ (m_{ab}) \in \F[t]^{n \times n}\,\big|\, m_{ab}(0)=0
   \text{ for all } 1\leq b<a\leq n \big\}.
\]
Notice that~$\cM$ consists exactly of all matrices where the elements below the
diagonal are multiples of~$t$.
It is easy to see that~$\cM$ is a (non-commutative) subring of
$\F[t]^{n \times n}$.
\begin{prop}\label{prop:xi}
Let $\Azs$ and~$\cM$ be as above.
Then the map
\begin{align}
     \xi:\qquad\qquad\Azs\qquad\qquad&\longrightarrow\qquad\qquad\qquad\qquad\cM  \nonumber\\[1ex]
     \sum_{l=0}^N z^{nl} \sum_{i=0}^{n-1} z^i [c_{i1}^l, c_{i2}^l, \dots , c_{in}^l]
    &\longmapsto\sum_{l=0}^N t^l \left( \begin{array}{ccccc}
                    c_{01}^l     & c_{12}^l  & \cdots & \cdots         & c_{n-1,n}^l \\
                    tc_{n-1,1}^l & c_{02}^l  & \cdots  & \cdots        & c_{n-2,n}^l \\
                    \vdots       &  \ddots    &\ddots &                & \vdots \\
                    tc_{21}^l    &           & \ddots &\ddots          & c_{1n}^l \\
                    tc_{11}^l    & tc_{22}^l & \cdots &\!\!tc_{n-1,n-1}^l\!\! & c_{0n}^l
                    \end{array} \right) \label{eq:xi}
\end{align}
is a ring isomorphism.
\end{prop}
The identification of $\Azs$ and $\mathcal{M}$ was first shown in~\cite{Jat71},
where it has been studied for the more general situation of skew polynomial
rings over arbitrary semisimple rings with a monomorphism~$\sigma$.
This general situation has also been used in~\cite{EGPS07} in order to classify
certain group convolutional codes.
Our choice of the semisimple ring~$A$ and the automorphism~$\sigma$ leads to a
particularly simple proof which we would like to briefly present.

\begin{proof}
It is obvious that each polynomial in~$\Azs$ has a unique representation as on
the left hand side of\eqnref{eq:xi}.
Thus the map~$\xi$ is well-defined. Moreover, it is obvious that~$\xi$ is
bijective and additive, and it remains to show that it is multiplicative.
In order to do so, we firstly observe that
$\xi(\alpha)=\textrm{diag} (\alpha_1,\dots,\alpha_n)$ for any
$\alpha=[\alpha_1,\dots,\alpha_n]\in A$, and secondly, that
\[
   \xi(z^{nl+i})=t^l\begin{pmatrix}0&I_{n-i}\\ tI_{i}&0\end{pmatrix}
   =\begin{pmatrix}0&I_{n-1}\\ t&0\end{pmatrix}^{nl+i}
   \text{ for all }l\in\N_0
   \text{ and }i=0,\ldots,n-1.
\]
Now the identities
$\xi(\alpha\beta)= \xi(\alpha)\xi(\beta),\
  \xi(z^{\nu+\mu})= \xi(z^\nu) \xi(z^\mu)$,
and
$\xi(z^\nu \alpha)=\xi(z^\nu) \xi(\alpha)$, where $\nu,\,\mu\in\N_0$ and
$\alpha,\,\beta \in A$, are obvious, and one easily verifies
$\xi(\alpha)\xi(z)=\xi(z) \xi(\sigma(\alpha))$ for all $\alpha\in A$.
Using the additivity of $\xi$ we conclude that~$\xi$ is indeed multiplicative,
hence a ring isomorphism.
\end{proof}
For later purposes it will be handy to have an explicit formula for the entries
of $\xi(g)$.
For
\begin{equation}\label{eq:g}
  g:=\sum_{l=0}^N z^{nl} \sum_{i=0}^{n-1} z^i
                             [c_{i1}^l, c_{i2}^l, \dots , c_{in}^l]
\end{equation}
and $\xi(g)=(m_{ab})$ one computes
\begin{equation}\label{eq:mab}
  \mbox{}\hspace*{-.3em}
   m_{ab}=\sum_{l=0}^N t^{l+\sgn(a-b)}c_{b-a+\sgn(a-b)n,b}^l\ \text{ for }
   1\leq a,\,b\leq n,
\end{equation}
where
\[
  \sgn(x):=\left\{
  \begin{array}{ll}1&\text{if }x> 0,\\ 0&\text{else}.\end{array}\right.
\]
It is clear that the subring $\xi(\F[z])$ of~$\cM$ is given by the set of
matrices

\begin{equation}\label{eq:xif}
  \sum_{l=0}^Nt^l\left( \begin{array}{ccccc}
                    c_{0}^l     & c_{1}^l  & \cdots & \cdots         & c_{n-1}^l \\
                    tc_{n-1}^l & c_{0}^l  & \cdots  & \cdots        & c_{n-2}^l \\
                    \vdots       &  \ddots    &\ddots &                & \vdots \\
                    tc_{2}^l    &           & \ddots &\ddots          & c_{1}^l \\
                    tc_{1}^l    & tc_{2}^l & \cdots &\!\!tc_{n-1}^l\!\! & c_{0}^l
                    \end{array} \right)\text{ where }N\in\N_0,\, c_i^l\in\F,
\end{equation}
and, defining $f\cdot M:=\xi(f)M$ for $f\in\F[z]$ and $M\in \cM$,
we may impose a left $\F[z]$-module structure on~$\cM$ making it isomorphic
to~$\Azs$ as left $\F[z]$-modules.
Due to the form of~$\xi(f)$ as given in\eqnref{eq:xif} the thus obtained module
structure is of course not identical to the canonical $\F[t]$-module structure
of~$\cM$.

In the sequel we will translate various properties of polynomials in~$\Azs$
into matrix properties in~$\cM$.
First of all we will show how to identify the units in~$\cM$.
This will play an important role later on when discussing left ideals that are
direct summands.
\begin{prop}\label{prop:unitsM}
Let $\cM^\times$ be the group of units of $\cM$. Then $\cM^\times = \cM \cap GL_n(\F[t])$.
In other words, $M\in\cM$ is a unit in the ring~$\cM$ if and only if
$\det(M)\in\F^\times:=\F \setminus \{0 \}$.
\end{prop}
As a consequence, left- (or right-) invertible elements in~$\cM$ or in
$\Azs$ are units.

\begin{proof}
The inclusion ``$\subseteq$'' is clear.
For the inclusion ``$\supseteq$'' let $M\in\cM\cap GL_n(\F[t])$.
Then there exists a matrix $N\in GL_n(\F[t])$ such that $MN=I$.
Substituting $t=0$ we obtain $M(0)N(0)=I$, thus $N(0)=M(0)^{-1}$.
Since $M\in\cM$ the matrix $M(0)$ is upper triangular.
But then the same is true for $N(0)$, showing that $N\in\cM$.
\end{proof}

Let us now turn to properties of the skew polynomial ring $\Azs$ that follow
from the semi-simplicity of~$A$.
Since the idempotents $e_1,\ldots,e_n$ are pairwise orthogonal and satisfy
$e_1+\ldots+e_n=1$ we have
\[
    \Azs = \rideal{e_1}\oplus \cdots \oplus \rideal{e_n} =
           \lideal{e_1} \oplus \cdots \oplus \lideal{e_n}.
\]
As a consequence, each element $g\in\Azs$ has a unique decomposition
\begin{equation}\label{eq:Peirce}
   g=g^{(1)}+\ldots+g^{(n)},\text{ where }g^{(a)}:=e_ag.
\end{equation}
The following notions will play a central role.
Recall that coefficients of skew polynomials are always meant to be right hand side
coefficients.
\begin{defi}
Let $g \in \Azs$.
\begin{alphalist}
\item For $a=1,\ldots,n$ we call $g^{(a)}:=e_a g$ the $a$-th {\em component}
      of~$g$.
\item The {\em support} of~$g$ is defined as $T_g := \set{a\mid g^{(a)} \ne 0}$.
\item We call $g$ {\em delay-free} if $T_g=T_{g_0}$ where $g_0\in A$ is the constant
      term of the polynomial~$g$.
\item The polynomial~$g$ is said to be {\em semi-reduced} if the leading
      coefficients of the components $g^{(a)}, \, a\in T_g$, lie in pairwise
      different ideals $\ideal{e_{i_a}}$.
\item The polynomial~$g$ is called {\em reduced\/} if no leading term of any
      component of~$g$ is a right divisor of a term of any other component
      of~$g$.
\end{alphalist}
\end{defi}
In order to comment on these notions let us have a closer look at the
components of a polynomial.
Using\eqnref{eq:sigmae} one obtains that, for instance, the first component of
a polynomial~$g$ is of the form
\[
  g^{(1)}=[c_0,0,\dots,0]+z[0,c_1,0,\dots,0]+z^2[0,0,c_2,0,\dots,0]+\cdots
   =c_0e_1+zc_1e_2+z^2c_2e_3+\cdots
\]
for some $c_j\in\F$.
In general we derive from\eqnref{eq:sigmae}
\begin{equation}\label{eq:gi}
    g^{(a)}=\sum_{j=0}^{N_a} z^jc_{a,j}e_{(a+j-1\text{ mod } n)+1}
    \text{ for some }N_a\in\N_0\text{ and }c_{a,j}\in\F
\end{equation}
for $a=1,\ldots,n$.
In particular, the coefficients of the components are $\F$-multiples of
primitive idempotents.
As a consequence,~$g$ is semi-reduced if and only if no leading term of any
component of~$g$ is a right divisor of the leading term of any other component
of~$g$.
Obviously, reducedness implies semi-reducedness, and a polynomial consisting
of one component is always reduced.

The concept of reducedness has been introduced for the skew polynomial ring~$\Azs$
in~\cite[Def.~4.9]{GS04}.
It proved to be very useful in the theory of CCC's.
In particular, it has led to the following results concerning minimal encoder
matrices and Forney indices of CCC's, see
\cite[Thm.~4.5, Thm.~4.15, Prop.~7.10, Thm.~7.13(b)]{GS04}.
\begin{theo} \label{T-ideals}
As before let $n\mid(q-1)$.
Let~$\sigma$ be any automorphism in $\AutF(A)$ and consider the
identification~$\p$ of $\F[z]^n$ with the skew polynomial ring
$\Azs$ as given in\eqnref{eq:p} and where we identify~$A$ with the quotient ring
$\F[x]/\ideal{x^n-1}$.
\begin{arabiclist}
\item For every $g\in\Azs$ there exists a unit $u\in\Azs^\times$ such that $ug$
      is reduced.
\item Let $\cC\subseteq\F[z]^n$ be a $\sigma$-CCC.
      Then there exists a reduced and delay-free polynomial $g\in\Azs$ such
      that $\p(\cC)=\lideal{g}:=\{fg\mid f\in\Azs\}$.
      In particular, the left ideal~$\p(\cC)$ is principal.
      Moreover, the polynomial~$g$ is unique up to left multiplication by units
      in~$A$.
\item Let $g\in\Azs$ be a reduced polynomial.
      Then $\p^{-1}(\lideal{g})\subseteq\F[z]^n$ is a direct summand of
      $\F[z]^n$ if and only if $g=\sum_{l\in T_g}u^{(l)}$ for some unit
      $u\in\Azs$.
      That is, the $\sigma$-CCC's are obtained by taking
      any unit~$u$ in $\Azs$ and choosing any collection of components of~$u$
      that forms a reduced polynomial.
\item Let $g\in\Azs$ be a reduced polynomial with support
      $T_g=\{i_1,\ldots,i_k\}$.
      Then $\cC:=\p^{-1}(\lideal{g}) =\im G=\{uG\mid u\in\F[z]^k\}$, where
      \begin{equation}\label{eq:G}
          G:=\begin{pmatrix} \p^{-1}(g^{(i_1)})\\ \vdots\\ \p^{-1}(g^{(i_k)})
             \end{pmatrix}\in\F[z]^{k\times n}.
      \end{equation}
      Furthermore, if~$\cC$ is a direct summand of~$\F[z]^n$ then~$G$ is a
      minimal encoder matrix of~$\cC$.
      As a consequence,~$\cC$ is a $\sigma$-CCC of dimension~$k$ with Forney
      indices $\deg g^{(i_1)},\ldots,\deg g^{(i_k)}$
      and degree $\delta:=\sum_{l=1}^k\deg g^{(i_l)}$.
\end{arabiclist}
\end{theo}

Let us now return to the situation where~$\sigma$ is as in\eqnref{eq:sigma}.
Throughout this paper semi-reduced polynomials will be much more handy than
reduced ones, see Proposition~\ref{prop:DegMat}(4) below.
Fortunately, it can easily be confirmed that the weaker notion of semi-reducedness
is sufficient for the results above to be true.
We confine ourselves to presenting the following details.

\begin{rem}\label{R-CCCsemireduced}
The results of Theorem~\ref{T-ideals}(1)~--~(4), except for the uniqueness
result in~(2), are true for semi-reduced polynomials as well.
In order to see this, one has to confirm that, firstly, all results of
\cite[Section~4]{GS04}, in particular Proposition~4.10, Corollary~4.13, and
Lemma~4.14, with the exception of the uniqueness result in Theorem~4.15,
remain true for semi-reduced polynomials.
Secondly one can easily see that Theorem~7.8, Proposition~7.10, Corollary~7.11,
Theorem~7.13, and Corollary~7.15 of \cite[Section~7]{GS04} remain true if one
replaces reducedness by semi-reducedness.
In all cases the proofs in~\cite{GS04} remain literally the same.
\end{rem}

Due to part~(3) of the theorem above the following notion will be important to
us.
\begin{defi}\label{D-unimodularg}
A polynomial $g\in\Azs$ is called {\em basic\/} if there exists a unit
$u\in\Azs^\times$ such that $g=\sum_{l\in T_g}u^{(l)}$.
\end{defi}
According to part~(3) and~(4) of Theorem~\ref{T-ideals}, see also
Remark~\ref{R-CCCsemireduced}, a semi-reduced polynomial~$g$ is basic if and
only if the matrix~$G$ in\eqnref{eq:G} is basic.

Let us now have a closer look at semi-reduced polynomials.
From\eqnref{eq:gi} we see immediately that
\begin{equation}\label{eq:semireduced}
 \text{$g$ is semi-reduced }\Longleftrightarrow
 \left\{\begin{array}{l}
 \text{the numbers $(a+\deg g^{(a)} -1 \text{ mod } n)+1,\ a\in T_g,$}\\[.6ex]
 \text{are pairwise different.}\end{array}\right.
\end{equation}
Moreover, one has
\begin{equation}\label{eq:semireduced2}
  g\in\Azs^\times\text{ and } g\text{ semi-reduced}\Longrightarrow
  g\in A^\times.
\end{equation}
Indeed, suppose $1=hg=(he_1+\ldots+he_n)g=\sum_{l=1}^n he_lg$.
Since $\ideal{e_l}\cong\F$ we get $\deg(he_lg)=\deg(he_l)+\deg(e_lg)$ and the leading
coefficient of $he_lg$ is in the same ideal as the leading coefficient
of $g^{(l)}$.
Now semi-reducedness of~$g$ shows that no cancelation of the leading
coefficients in $\sum_{l=1}^n he_lg$ is possible and thus $hg=1$ implies that
$\deg(g^{(l)})=0$ for all $l=1,\ldots,n$.

It is easy to translate these notions into the setting of the matrix ring~$\cM$.
Indeed, defining the standard basis matrices
\begin{equation}\label{eq:Eij}
  E_{ab}\in\F^{n\times n} \text{ via }\big(E_{ab}\big)_{ij}=
  \left\{\begin{array}{ll} 1 &\text{if }(i,j)=(a,b)\\ 0&\text{if }(i,j)\not=(a,b)
  \end{array}\right\}\text{ for }a,b=1,\ldots,n
\end{equation}
we obtain
\begin{equation}\label{eq:xiei}
  \xi(e_a)=E_{aa}.
\end{equation}
Thus,
\begin{equation}\label{eq:compM}
   \xi(g^{(a)})=\xi(e_a)\xi(g)=E_{aa}\xi(g)
\end{equation}
is a matrix where at most the $a$-th row is nonzero.
It is obtained from~$\xi(g)$ by deleting all other rows.
This gives rise to the following definition.
\begin{defi}
Let $M=(m_{ab}) \in \cM$. Then we define the support of~$M$ to be
\[
    \Supp(M) :=\{a\mid\text{the $a$-th row of $M$ is non-zero}\}.
\]
We say that~$M$ is {\em delay-free\/} if $m_{aa}(0) \ne 0$ for all
$a \in \Supp(M)$.
\end{defi}
By definition $\Supp\big(\xi(g)\big)=T_g$ for all $g\in\Azs$.
Furthermore, the very definition of~$\xi$ shows that $g$ is delay-free if and
only if $\xi(g)$ is.
Finally, we have the implication
\begin{equation}\label{eq:rk M}
   \text{$g$ is delay-free }\Longrightarrow
   \rank_{\F[t]}\xi(g)=\rank_{\F}\xi(g)|_{t=0}=|T_g|.
\end{equation}
This follows by observing that the nonzero rows of~$\xi(g)|_{t=0}$ form a
matrix in row echelon form with pivot positions in the columns with indices
in $T_g$.

In order to express semi-reducedness in terms of the matrix $\xi(g)$ we need
the following concept.
\begin{defi}\label{D-DegMat}
Let $M=\big(m_{ab}\big)\in\cM$ and put $d_{ab}:=\deg(m_{ab})$ for $a,b=1,\ldots,n$
(where, as usual, the zero polynomial has degree $-\infty$).
The {\em degree matrix\/} of~$M$ is defined as
\[
  \cD(M)= \begin{pmatrix}
            n d_{11}      & n d_{12} +1 & n d_{13} +2 & \cdots         & n d_{1n}+n-1 \\
            n d_{21} -1   & n d_{22}    & n d_{23} +1 & \cdots         & n d_{2n}+n-2 \\
            n d_{31} -2   & n d_{32} -1 & n d_{33}    & \cdots         & n d_{3n}+n-3 \\
            \vdots        & \vdots      & \vdots      &  \ddots          & \vdots\\
             n d_{n1} -n+1 & n d_{n2} -n+2&\ n d_{n3} -n+3\ & \cdots        & n d_{nn}
            \end{pmatrix}.
\]
Hence $\cD(M)_{ab}=nd_{ab}-a+b\in\N_0\cup\{-\infty\}$ for $a,b=1,\ldots,n$.
We call a row of $\cD(M)$ {\em trivial\/} if all entries are $-\infty$.
The matrix~$M$ is said to be {\em semi-reduced\/} if the maxima in the
non-trivial rows of $\cD(M)$ appear in different columns.
\end{defi}
Obviously, the trivial rows of $\cD(M)$ correspond to the zero rows of~$M$.
Furthermore, we have the following properties.
\begin{prop}\label{prop:DegMat}
Let $g\in\Azs$ and $M=\xi(g)=(m_{ab})$.
For $a,b=1,\ldots,n$ put $d_{ab}=\deg m_{ab}$.
\begin{arabiclist}
\item In each row and column of~$\cD(M)$ the entries different from~$-\infty$
      are pairwise different.
      In particular, each non-trivial row has a unique maximum.
\item For all $a,b=1,\ldots,n$ we have
      $\xi\big(g^{(a)}e_b\big)=E_{aa}M E_{bb}=m_{ab}E_{ab}$ and
      $\deg\big(g^{(a)}e_b\big)=\cD(M)_{ab}$.
      Furthermore,
      \begin{equation}\label{eq:degga}
          \deg g^{(a)}=\max_{1\le b\le n}\cD(M)_{ab}
      \text{ for all }a=1,\ldots,n
      \end{equation}
      and
      \[
         \deg g= \max_{1\le a\le n} (\deg g^{(a)})
               = \max_{1 \le a,b \le n} \cD(M)_{ab}.
       \]
\item $g$ is semi-reduced if and only if~$M$ is semi-reduced.
\item For $a\in\Supp(M)$ let $\delta_a:=\max\{d_{ab}\mid 1\leq b\leq n\}$ and
      put $b_a:=\max\{b\mid d_{ab}=\delta_a\}$; that is, $m_{a,b_a}$ is the
      rightmost entry in the $a$-th row of~$M$ having maximal degree $\delta_a$.
      Then $\max_{1\leq b\leq n}\cD(M)_{ab}=\cD(M)_{a,b_a}$.
      As a consequence,~$M$ is semi-reduced if and only if the indices
      $b_a,\ a\in\Supp(M)$, are pairwise different.
\end{arabiclist}
\end{prop}
\begin{proof}
Part~(1) is trivial.
The first assertion of~(2) follows from\eqnref{eq:xiei} and the multiplicativity
of~$\xi$.
As for the degree of $g^{(a)}e_b$ let~$g$ be as in\eqnref{eq:g}.
Using\eqnref{eq:sigmae} we obtain
\[
  e_ag=\sum_{l\geq0}\sum_{i=0}^{n-1}z^{nl+i}c^l_{i,(i+a-1\mod n)+1}
                                  e_{(i+a-1\mod n)+1}.
\]
Notice that for any $b=1,\ldots,n$ we have
\[
  e_{(i+a-1\mod n)+1}e_b=e_b\Longleftrightarrow(i+a-1\mod n)+1=b
  \Longleftrightarrow i=b-a+\sgn(a-b)n
\]
while $e_{(i+a-1\mod n)+1}e_b=0$ for all other values of~$i\in\{0,\ldots,n-1\}$.
Thus
\[
 g^{(a)}e_b=e_age_b=\sum_{l\geq 0}z^{n\big(l+\sgn(a-b)\big)+b-a}
                                        c^l_{b-a+\sgn(a-b)n,b}e_b.
\]
As a consequence,
\[
  \deg(g^{(a)}e_b)=\Big(\!\!\max\{l\mid c^l_{b-a+\sgn(a-b)n,b}\not=0\}
                     +\sgn(a-b)\Big)n+b-a
     =n\deg(m_{ab})+b-a,
\]
where the last identity follows from\eqnref{eq:mab}.
This shows $\deg(g^{(a)}e_b)=\cD(M)_{ab}$.
The last two statements of~(2) are a direct consequence.
For parts~(3) and~(4) let~$\delta_a$ and~$b_a$ be as in~(4).
Then
$\max\{\cD(M)_{ab}\mid 1\leq b\leq n\}
 =\max\{n d_{ab}-a+b\mid 1\leq b\leq n\}
 =n\delta_a-a+b_a=\cD(M)_{a,b_a}$.
Now the rest of~(4) is a consequence of Definition~\ref{D-DegMat}, whereas~(3)
follows from\eqnref{eq:semireduced} along with part~(2) since
$a+\deg g^{(a)}-1\mod n=b_a-1$.
\end{proof}

Let us consider some examples.
\begin{exa}\label{Exa:CCCMDS}
Let~$\alpha$ be a primitive element of~$\F$ and put $n=q-1$.
Then $x^n-1=\prod_{i=0}^{n-1}(x-\alpha^i)$.
Consider the isomorphism
$\phi: A'\longrightarrow A,\ f\longmapsto
  [f(1),f(\alpha),\ldots,f(\alpha^{n-1})]$
of $A':=\F[x]/{\ideal{x^n-1}}$ with~$A$.
Then the primitive idempotents of~$A'$ are given by
$e_a:=\gamma_a\prod_{i\not=a-1}(x-\alpha^i)$ for $a=1,\ldots,n$ and some
constants $\gamma_a\in\F^\times$.
Choose now the automorphism~$\sigma\in\AutF(A')$ defined via
$\sigma(x)=\alpha^{-1} x$.
Then one easily checks that\eqnref{eq:sigmae} is satisfied.
This example has also been studied in \cite[Prop.~4.2]{GL06p} (since in that
paper the automorphism is given by $x\mapsto\alpha x$ one has to
replace~$\alpha$ by~$\alpha^{-1}$ in the idempotents in order to get
back the ordering as in~\cite{GL06p}).
In \cite[Prop.~4.2, Thm.~2.1]{GL06p} it has been shown that for any
$1\leq\delta\leq n-1$ the $\sigma$-cyclic submodule $\cC=\p^{-1}(\lideal{g})$,
where $g=e_1\sum_{i=0}^{\delta}z^i$, gives rise to a 1-dimensional MDS
convolutional code.
Hence the distance of that code is $n(\delta+1)$, which is the maximum value
among all 1-dimensional codes of length~$n$ and degree~$\delta$
(for MDS convolutional codes see~\cite{RoSm99}).
Let us compute an encoder matrix~$G\in\F[z]^{1\times n}$ of~$\cC$.
From Theorem~\ref{T-ideals}(4) we know that $G=\p^{-1}(g)$ is such an encoder
since $g=g^{(1)}$.
The first primitive idempotent in~$A'$ is given by
$e_1=\frac{1}{n}\sum_{i=0}^{n-1}x^i$ as one can easily verify via the
isomorphism~$\phi$.
Thus,
\[
  g=\frac{1}{n}\sum_{i=0}^{n-1}x^i\sum_{j=0}^{\delta}z^j
   =\frac{1}{n}\sum_{j=0}^{\delta}z^j\sum_{i=0}^{n-1}\sigma^j(x^i)
   =\frac{1}{n}\sum_{j=0}^{\delta}z^j\sum_{i=0}^{n-1}\alpha^{-ji}x^i.
\]
Using the map~$\p$ from\eqnref{eq:p} we obtain
\[
   G=\frac{1}{n}\Big(\sum_{j=0}^{\delta}z^j,\sum_{j=0}^{\delta}z^j\alpha^{-j},
                   \ldots,\sum_{j=0}^{\delta}z^j\alpha^{-(n-1)j}\Big)
    =\frac{1}{n}\sum_{j=0}^{\delta}z^j
           (1,\,\alpha^{-j},\,\alpha^{-2j},\ldots,\, \alpha^{-(n-1)j}).
\]
In this specific example the matrix $M:=\xi(g)$ has a particularly simple form.
Indeed, since $g=g^{(1)}=\sum_{i=0}^{\delta}z^ie_{1+i}$ and $\delta\leq n-1$ we
obtain
\begin{align*}
   &\mbox{}\hspace*{6.7em}\overbrace{\hspace*{2.8em}}^{\delta+1}\\[-1.3ex]
   &M:=\xi(g)=\left(
      \begin{smallmatrix} 1&1&\ldots&1&0&\ldots&0\\[.5ex]
                             \ \\ & & & \mbox{\huge 0} \\[.5ex] \
      \end{smallmatrix}\right)
   \text{ and }\cD(M)=\left(
       \begin{smallmatrix}0&1&\ldots\,\ \delta\ -\infty&\ldots&-\infty\\[.5ex]
           \ \\  & &  \mbox{\huge$-\infty$} \\[.6ex] \
       \end{smallmatrix}\right).
\end{align*}
The degree equation\eqnref{eq:degga} is obvious.
Notice that the matrix~$M$ is idempotent.
Thus, $g$ is an idempotent generator of the left ideal~$\lideal{g}$.
\end{exa}

\begin{exa}\label{Exa:reduction}
Let $q=5$ and $n=4$. Consider the matrix
\[
     M=\begin{pmatrix} 2+t&1&4&4\\0&1&3&0\\4t&2t&1+t&1\\0&0&0&0\end{pmatrix}
     \in\cM.
\]
Obviously,~$M$ is delay-free.
The corresponding polynomial $g=\xi^{-1}(M)\in\Azs$ is given by
\[
  g=2e_1+e_2+e_3+z(e_2+3e_3+e_4)+z^2(4e_1+4e_3)+z^3(2e_2+4e_4)+z^4(e_1+e_3),
\]
and its components are
\[
   g^{(1)}=2e_1+ze_2+4z^2e_3+4z^3e_4+z^4e_1,\ g^{(2)}=e_2+3ze_3,\
   g^{(3)}=e_3+ze_4+4z^2e_1+2z^3e_2+z^4e_3.
\]
The matrix, and thus~$g$, are not semi-reduced as we can see directly from
Proposition~\ref{prop:DegMat}(4) or from the degree matrix
\[
   \cD(M)=\begin{pmatrix}4&1&2&3\\-\infty&0&1&-\infty\\2&3&4&1\\
   -\infty&-\infty&-\infty&-\infty\end{pmatrix}.
\]
Using two steps of elementary row operations one can bring~$M$ into
semi-reduced form without changing the left ideal~$\lideal{M}$.
Indeed, one easily checks that
\begin{equation}\label{eq:barM}
  \begin{pmatrix}1&0&1&0\\0&1&0&0\\0&0&1&0\\0&0&0&1\end{pmatrix}
  \begin{pmatrix}1&0&0&0\\0&1&0&0\\0&3t&1&0\\0&0&0&1\end{pmatrix}
  \begin{pmatrix} 2+t&1&4&4\\0&1&3&0\\4t&2t&1+t&1\\0&0&0&0\end{pmatrix}
  =\begin{pmatrix} 2&1&0&0\\0&1&3&0\\4t&0&1&1\\0&0&0&0\end{pmatrix}=:\bar{M}.
\end{equation}
Since the two leftmost matrices are in~$\cM^\times$, see
Proposition~\ref{prop:unitsM}, we have $\lideal{M}=\lideal{\bar{M}}$.
Furthermore,~$\bar{M}$ is semi-reduced as one can easily see with the help of
Proposition~\ref{prop:DegMat}(4).
The corresponding polynomial~$\bar{g}:=\xi^{-1}(\bar{M})$ is given by
$\bar{g}=2e_1+e_2+e_3+z(e_2+3e_3+e_4)+z^24e_1$, which is semi-reduced according
to Proposition~\ref{prop:DegMat}(3).
We will come back to this example later on in Example~\ref{Exa:reduction2}.
\end{exa}

The procedure of semi-reducing matrices in~$\cM$ as in the previous example
always works.
This is shown in the following result.

\begin{theo}\label{thm:Munits}
A matrix in~$GL_n(\F[t])$ is said to be an elementary unit of~$\cM$ if it is of
any of the following types:
\begin{romanlist}
\item $\sum_{i\not=a}E_{ii}+\alpha E_{aa}$ for some $\alpha\in\F^\times$ and
      $a=1,\ldots,n$,
\item $I_n+t^N\alpha E_{ab}$ for some $N\geq 0,\; \alpha\in\F$ and
      $1\leq a<b\leq n$,
\item $I_n+t^N\alpha E_{ab}$ for some $N>0,\;\alpha\in\F$ and $1\leq b<a\leq n$.
\end{romanlist}
Then each matrix in~$\cM$ can be brought into semi-reduced form via left
multiplication by finitely many elementary units of~$\cM$.
As a consequence, each unit $M\in\cM^\times$ can be written as a product of
elementary units of~$\cM$.
\end{theo}
\begin{proof}
It is trivial that each of the matrices in (i) --~(iii) is indeed
in~$\cM^\times$, and that the inverse of such an elementary unit is an
elementary unit of the same type.
\\
Let now $M=(m_{ab})\in\cM$ and assume that~$M$ is not semi-reduced.
Then, according to Proposition~\ref{prop:DegMat} there exist indices $(a,c)$
and $(b,c)$ where $b>a$ such that $\cD(M)_{a,c}$ and $\cD(M)_{b,c}$ are the
maxima in the $a$-th and $b$-th row of~$\cD(M)$, respectively.
Using the definition of the matrix $\cD(M)$ one easily checks that for $b>a$
and $j=1,\ldots,n$
\begin{equation}\label{eq:DMentries}
   \cD(M)_{bj}< \cD(M)_{aj}\Longleftrightarrow \deg(m_{bj})\leq \deg(m_{aj}).
\end{equation}
We want to transform~$M$ into semi-reduced form via left multiplication by
elementary units of~$\cM$.
Consider first the case $\cD(M)_{bc}\leq\cD(M)_{ac}$.
Put $\hat{M}=(\hat{m}_{ij})=(I_n+t^N\alpha E_{ab})M$ where
$N:=\deg(m_{ac})-\deg(m_{bc})\geq 0$ and $\alpha\in\F$ is such that
$\deg(\hat{m}_{ac})=\deg(m_{ac}+t^N\alpha m_{bc})<\deg(m_{ac})$.
This is possible due to\eqnref{eq:DMentries}.
Then $\cD(\hat{M})_{ac}<\cD(M)_{ac}$.
Furthermore, for $j\not=c$ we have
$\hat{m}_{aj}=m_{aj}+t^{N}\alpha m_{bj}$ and thus
\begin{align*}
  \cD(\hat{M})_{aj}&=n\deg(\hat{m}_{aj})-a+j
                   \leq n\max\{\deg(m_{aj}),N+\deg(m_{bj})\}-a+j\\
      &=\max\{n\deg(m_{aj})-a+j,nN+n\deg(m_{bj})-a+j\}\\
      &=\max\{\cD(M)_{aj},\cD(M)_{bj}+nN+b-a\}\\
      &<\max\{\cD(M)_{ac},\cD(M)_{bc}+nN+b-a\},
\end{align*}
where the last inequality holds true due to the uniqueness of the row maxima
in $\cD(M)$ (see Proposition~\ref{prop:DegMat}(1)).
Since
\[
  \cD(M)_{bc}+nN+b-a=n\deg(m_{bc})-b+c+nN+b-a=n\deg(m_{ac})-a+c=\cD(M)_{ac}
\]
the above results in $\cD(\hat{M})_{aj}< \cD(M)_{ac}$ for all $j=1,\ldots,n$.
\\
If $\cD(M)_{bc}>\cD(M)_{ac}$, then put $\hat{M}:=(I_n+t^N\alpha E_{ba})M$ where
$N> 0$ and $\alpha\in\F$ are chosen such that
$\deg(m_{bc}+t^N\alpha m_{ac})<\deg(m_{bc})$.
Again, this is possible by virtue of\eqnref{eq:DMentries}.
Arguing the same way we obtain likewise $\cD(\hat{M})_{bj}< \cD(M)_{bc}$ for
all $j=1,\ldots,n$.
Summarizing we see that if~$M$ is not semi-reduced we may apply one of the two
transformations given above, and they both strictly decrease the maximum value
in one of the rows of~$\cD(M)$ while all other rows remain unchanged.
Altogether this results in a reduction procedure that must stop after finitely
many steps with a semi-reduced matrix.
\\
The last statement of the theorem follows directly from\eqnref{eq:semireduced2}
together with the use of elementary units of Type~(i).
\end{proof}

One easily checks that the elementary units in Theorem~\ref{thm:Munits}
correspond to the units
\[
  \xi^{-1}\big(\sum_{i\not=a}E_{ii}+\alpha E_{aa}\big)
    =\alpha e_a+\sum_{i\not=a}e_i,\quad
  \xi^{-1}\big(I_n+t^N\alpha E_{ab}\big)=1+z^{Nn+b-a}\alpha e_b,
\]
in~$\Azs$.
Of course, in the second case we have $a\not=b$ and $N>0$ if $b<a$.
These units in~$\Azs$ have been studied in detail in the paper~\cite{GL06}.
In a more general context it has been shown in~\cite[Lem.~3.7]{GL06} that
$1+z^d\alpha e_b$, where $d>0$, is a unit if and only if $n\nmid d$, and that
each unit in~$\Azs$ can be written as the product of finitely many units of the
types above.
This corresponds exactly to our last statement in Theorem~\ref{thm:Munits} above.
With that theorem we even see that left reduction to a semi-reduced polynomial
in the skew-polynomial ring~$\Azs$, cf.~\cite[Cor.~4.13]{GS04}, simply means
elementary row reduction of the matrices in~$\cM$.

The following result tells us as to when a matrix in~$\cM$ can be completed to
a unit in~$\cM$ by filling in suitable entries in the zero rows.
This result will be important later on when studying whether a cyclic submodule
is a direct summand.
\begin{lemma} \label{lem:replacing 0 rows}
Let $M\in \cM$ be delay-free and $\Supp(M)=\set{i_1, \ldots, i_k}$, where
$i_1<\ldots<i_k$.
Denote the rows of~$M$ by $M_1,\ldots, M_n$.
Then the following are equivalent.
\begin{alphalist}
\item The matrix
      $\widetilde{M}:=\left( \begin{smallmatrix}M_{i_1}\\ \vdots\\ M_{i_k}
       \end{smallmatrix}\right)$ is basic.
\item There exist vectors $\widehat{M}_i\in\F[t]^{1\times n}$ for
      $i\in\{1,\ldots,n\}\backslash\{i_1,\ldots,i_k\}$ such that the matrix
      \[
        N=\begin{pmatrix}N_1\\ \vdots\\ N_n\end{pmatrix}\in\F[t]^{n\times n},
        \text{ where }
        N_i:=\left\{\begin{array}{ll}
           M_i&\text{if } i\in\{i_1,\ldots,i_k\}\\[.6ex]
           \widehat{M}_i&\text{if }
             i\in\{1,\ldots,n\}\backslash\{i_1,\ldots,i_k\}
             \end{array}\right.
      \]
      is in $\cM^\times$.
\end{alphalist}
\end{lemma}

Since every matrix in~$\cM^\times$ is delay-free it is clear that the
delay-freeness is necessary for the implication~(a) $\Longrightarrow$~(b) to be
true.

For the proof the following suggestive notation will be helpful.
Arbitrary entries of a matrix in $\F[t]^{r \times n}$ will be indicated by an
asterisk $\ast$, whereas entries that are multiples of~$t$ will be denoted by the
symbol $\ideal{t}$.
Thus, the elements of~$\cM$ are just the matrices of the form
\[
    \left(  \begin{smallmatrix}
            \ast   & \ast   & \cdots & \ast   \\
            \ideal{t}    & \ast   &        & \ast   \\
            \vdots & \ddots & \ddots & \vdots \\
            \ideal{t}    & \cdots & \ideal{t}    & \ast   \\
        \end{smallmatrix}   \right) .
\]

\begin{proof}
Only the implication ``(a)~$\Rightarrow$~(b)'' requires proof.
It is a well-known fact that, since the $k$-minors of $\widetilde{M}$ are
coprime, there exists a matrix $\widehat{M} \in \F[t]^{(n-k)\times n}$ such
that
$N:=\left( \begin{smallmatrix} \widetilde{M} \\ \widehat{M} \end{smallmatrix}
    \right) \in GL_n(\F[t])$.
Therefore, $N$ is of the form
\[
    \left(  \begin{smallmatrix}
        \ideal{t} & \cdots & m_{i_1 i_1} & * & \cdots &        &        &             &   &        & * \\
        \ideal{t} & \cdots & \cdots & \cdots & m_{i_2 i_2} & * & \cdots &             &   &        & * \\
        \vdots &     &        &        &        &        &        &             &   &   & \vdots \\
        \ideal{t} & \cdots & \cdots & \cdots & \cdots & \cdots & \cdots & m_{i_k i_k} & * & \cdots & * \\[.7ex] \hline\\[.7ex]
            &        &        &        & \widehat{M}& &   &             &   &        &   \\[.7ex]
        \end{smallmatrix}   \right),
\]
where all entries in the upper block and to the left of $m_{i_j,i_j}$ are of the type
$\ideal{t}$.
We will show now that with suitable row operations we can transform~$N$ into a
matrix being in~$\cM$ without altering the rows of the upper part~$\widetilde{M}$.
By delay-freeness of~$M$ we have $m_{i_j i_j}(0) \ne 0$ for all $j=1,\ldots,k$.
Hence, by adding suitable multiples of rows of~$\widetilde{M}$ to $\widehat{M}$
we can transform~$N$ into a matrix of the form
\begin{equation}\label{eq:N}
    \left(  \begin{smallmatrix}
        \ideal{t} & \cdots & m_{i_1 i_1} & * & \cdots &        &        &        &        &        & * \\
        \ideal{t} & \cdots & \cdots & \cdots & m_{i_2 i_2} & * & \cdots &        &        &        & * \\
        \vdots &     &        &        &        &        &        &        &        &   & \vdots \\
        \ideal{t} & \cdots & \cdots & \cdots & \cdots & \cdots & \cdots & m_{i_k i_k} & * & \cdots & * \\[.7ex]\hline\\[.7ex]
            &        & \ideal{t}    &        & \ideal{t}    &        &        & \ideal{t}    &        &        &   \\
        *   &        & \vdots & *      & \vdots & *      &        & \vdots &        & *      &   \\
            &        & \ideal{t}    &        & \ideal{t}    &        &        & \ideal{t}    &        &        &   \\
        \end{smallmatrix} \right) \in GL_n(\F[t]),
\end{equation}
where the upper part $\widetilde{M}$ is unaltered.
Therefore, we may assume without loss of generality that $N$ is as
in\eqnref{eq:N}.
For the next step observe that for any column vector $a\in\F[t]^{l \times 1}$
there exists $U\in GL_l(\F[t])$ such that $Ua$ is of the form
\[
    \left(  \begin{smallmatrix} * \\ \ideal{t} \\ \vdots \\ \ideal{t} \\
    \end{smallmatrix} \right).
\]
Thus, working consecutively from the left to the right and applying suitable
elementary row operations on the lower block of~$N$ we can bring
$\widehat{M}$ into the form

\begin{align*}
    &\widehat{M}=\left( \begin{array}{cccc|cccc|ccc}
    \ast   &        &        & \ideal{t}    &        &        &        & \ideal{t}    &        &        &  \\
    \ideal{t}& \ddots &      & \vdots       &        &        &        & \vdots &        &        &  \\
           & \ddots & \ast   & \ideal{t}    &        &        &        & \vdots &        &        &  \\\hline
    \ideal{t}&\cdots& \ideal{t}& \ideal{t}  & \ast   &    &        & \ideal{t}& \phantom{\vdots}    & &\\
    \vdots &        & \vdots & \vdots & \ideal{t}    & \ddots &        & \vdots &        &        &  \\
    \vdots &        & \vdots & \vdots &        & \ddots & \ast   & \ideal{t}    &       &       &  \\ \hline
    \vdots &        & \vdots & \vdots & \vdots &        & \ideal{t}    & \ideal{t}    & \ddots &        &  \\
    \vdots &        & \vdots & \vdots & \vdots &        & \vdots & \vdots &        & \ddots &  \\
    \ideal{t}    &\cdots & \ideal{t}    & \ideal{t}    & \ideal{t}    &        & \ideal{t}    & \ideal{t}    &        &        &  \\
    \end{array} \right)
    \begin{array}{l}
    \left.\begin{array}{l}\ \\[5.1ex] \\ \end{array}\hspace*{-1.6em}\right\}{\scriptstyle{i_1-1\text{ rows}}}\\[2.8ex]
    \left.\begin{array}{l}\ \\[5.1ex] \\ \end{array}\hspace*{-1.6em}\right\}{\scriptstyle{i_2-i_1-1\text{ rows}}}\\
    \phantom{\left.\begin{array}{l}\ \\[4.8ex] \\ \end{array}\right\}}
    \end{array}
    \\[-.7ex]
    &\hspace*{3.1em}\underbrace{\hspace*{7.6em}}_{i_1\text{ columns}}\quad
      \underbrace{\hspace*{7.6em}}_{i_2-i_1\text{ columns}}
\end{align*}
while~$\widetilde{M}$ does not change.
Assume that
$N=\left(\begin{smallmatrix}\widetilde{M}\\\widehat{M}\end{smallmatrix}\right)$
is of this form and recall that the $j$-th row of~$\widetilde{M}$ is given by
$M_{i_j},\,j=1,\ldots,k$.
Now, moving for $j=1,\ldots,k$ the $j$-th row of~$\widetilde{M}$ to the bottom of the
$j$-th block of~$\widehat{M}$ we can form a matrix $N'$ where the $i_j$-th row is
given by $M_{i_j}$ for $j=1,\ldots,k$ and where we keep the ordering of all
remaining rows of~$N$.
This way, the entries in~$\widehat{M}$ that are explicitly indicated by asterisks
and the entries $m_{i_j i_j}, j=1,\ldots,k,$ will appear on the diagonal of~$N'$,
while all entries below the diagonal will be in~$\ideal{t}$.
Hence $N'\in GL_n(\F[t])\cap\cM$ and Proposition~\ref{prop:unitsM} completes the proof.
\end{proof}

\begin{rem}\label{R-general}
Essentially all of the results of this section are true without the requirement
of~$n$ dividing $q-1$.
Indeed, just consider~$A$ and~$\sigma$ as in\eqnref{eq:A} and\eqnref{eq:sigma}.
Section~4 of~\cite{GS04} remains true in this case since it was solely based on~$A$
being a direct product of fields.
The only part that needs extra proof is part~(3) of Theorem~\ref{T-ideals}, and that
can be accomplished using ring theoretic methods.
Part~(4) of that theorem does not have a meaning in this more general setting
since~$A$ is not isomorphic to $\F[x]/{\ideal{x^n-1}}$ anymore, and thus the
map~$\p$ does not exist.
We will briefly come back to this situation in Section~\ref{S-general}.
\end{rem}

\section{Construction and Existence of $\sigma$-CCC's}\label{S-CCC}
\setcounter{equation}{0}
In this section we will apply the results obtained so far in order to construct
$\sigma$-CCC's, and we will discuss some existence issues.
As before, let~$A$ and~$\sigma$ be as in\eqnref{eq:A} and\eqnref{eq:sigma} and let
$n\mid(q-1)$.
Moreover, put
\begin{equation}\label{eq:Mbasic}
   \cMbasic:=\bigg\{M\in\cM \,\bigg|
   \begin{array}{l}\text{$M$ delay-free and}\\ \text{the nonzero rows of~$M$ form a basic matrix}\end{array}\bigg\}.
\end{equation}

Let us summarize the previous results in the following form.
We also think it is worthwhile pointing out that the property of a module being a direct
summand as an $\F[z]$-submodule is equivalent to being a direct summand as a left
ideal in the skew polynomial ring.
Recall the map~$\p$ from\eqnref{eq:p}.

\begin{theo}\label{T-sCCC}
Let $g\in\Azs$ be delay-free and semi-reduced. Then the following are equivalent.
\begin{romanlist}
\item $\lideal{g}$ is a direct summand of the left $\F[z]$-module $\Azs$ (thus, $\p^{-1}(\lideal{g})$ is a $\sigma$-CCC).
\item $\lideal{g}$ is a direct summand of the ring $\Azs$.
\item $\xi(g)\in\cMbasic$.
\end{romanlist}
\end{theo}
\begin{proof}
The equivalence of~(i) and~(ii) has been proved in \cite[Rem.~2.10]{GL06}, whereas
(i)~$\Longleftrightarrow$~(iii) follows from Lemma~\ref{lem:replacing 0 rows}, Theorem~\ref{T-ideals}(3), and the fact
that~$\xi$ is an isomorphism.
\end{proof}

Notice that~(iii) gives us an easy way of checking whether a given left ideal is a
direct summand (thus a $\sigma$-CCC) since basicness of a matrix in $\F[t]^{k\times n}$
can, for instance, be checked by testing whether its $k$-minors are coprime.

As a consequence, $\sigma$-CCC's are in one-one correspondence with the left ideals
$\lideal{M}\subseteq\cM$ where~$M\in\cMbasic$ is semi-reduced.
Proposition~\ref{prop:DegMat}(3) and Theorem~\ref{T-ideals}(4), see also
Remark~\ref{R-CCCsemireduced}, tell us immediately the algebraic parameters of the code,
that is, the dimension, the Forney indices and degree.
This is summarized in the next result.

\begin{cor} \label{cor:main corollary}
Let $M \in \cMbasic$ be semi-reduced and let $\Supp(M)=\{i_1,\ldots,i_k\}$, where
$i_1<\ldots <i_k$.
Then $\cC:=\p^{-1}(\lideal{\xi^{-1}(M)})\subseteq\F[z]^n$ is a $k$-dimensional
$\sigma$-CCC with Forney indices given by
$\nu_l:=\max_{1\leq j\leq n}\cD(M)_{i_l,j}$ for $l=1,\ldots,k$.
\end{cor}

\begin{exa}\label{Exa:2dimCCC}
Let us consider again the setting of Example~\ref{Exa:CCCMDS} where $n=q-1$ and
$\sigma(x)=\alpha^{-1}x$ with some primitive element~$\alpha$ of~$\F$.
In that example we presented some $1$-dimensional $\sigma$-cyclic MDS convolutional
codes.
These codes can be generalized as follows.
Let $g=(e_1+e_2)(1+z)$. Then
\[
    M=\xi(g)=\left(\begin{smallmatrix} 1&1&0&0&\ldots&0\\0&1&1&0&\ldots&0\\
                            \ \\ & & & \mbox{\huge 0}& \\ \ \end{smallmatrix}\right),
    \quad
    \cD(M)=\left(\begin{smallmatrix} 0&1&-\infty\ -\infty\ \ldots&-\infty\\0&0&\ \,1\quad -\infty\ \ldots&-\infty\\
                            \ \\ & & \mbox{\huge$-\infty$}& \\ \ \end{smallmatrix}\right).
\]
Obviously,~$M\in\cMbasic$ and hence the previous corollary guarantees that
$\cC:=\p^{-1}(\lideal{g})\subseteq\F[z]^n$ is a $\sigma$-CCC.
Furthermore,~$M$, and thus~$g$, is semi-reduced.
As a consequence, the code~$\cC$ is a $2$-dimensional code in $\F[z]^n$ with both Forney
indices equal to~$1$ and degree~$2$.
In other words,~$\cC$ is a unit memory code.
We will show now that these codes are optimal in the sense that they have the largest
distance among all $2$-dimensional codes with Forney indices $1,1$ and length~$n=q-1$
over~$\F_q$.
According to \cite[Prop.~4.1]{GS06p} (see also \cite[Eq.~(1.3)]{GL06})
the largest possible distance for codes with these
parameters, called the Griesmer bound, is given by the number $2(n-1)$.
In order to compute the actual distance of~$\cC$ we need an encoder matrix.
Since the support of~$g$ is $T_g=\{1,\,2\}$ and $g^{(1)}=e_1+ze_2,\,g^{(2)}=e_2+ze_3$,
Theorem~\ref{T-ideals}(4) tells us that a minimal encoder is given
by the matrix $G=G_0+G_1z$ where
\[
  G_0=\begin{pmatrix}\p^{-1}(e_1)\\\p^{-1}(e_2)\end{pmatrix},\
  G_1=\begin{pmatrix}\p^{-1}(e_2)\\\p^{-1}(e_3)\end{pmatrix}\in\F^{2\times n}.
\]
Furthermore,~$G$ is basic and minimal, and thus $\rank G_0=\rank G_1=2$.
Now it is easy to see that the two block codes generated by~$G_0$ and~$G_1$,
respectively, are MDS codes, that is, they both have distance~$n-1$.
Indeed, recall from Example~\ref{Exa:CCCMDS} that
$e_a=\gamma_a\prod_{i\not=a-1}(x-\alpha^i)$ for $a=1,\ldots,n$.
Thus, in the ring $A'=\F[x]/{\ideal{x^n-1}}$ the ideal $\ideal{e_1,\,e_2}$ is identical
to $\ideal{f}$, where $f=\prod_{i=2}^{n-1}(x-\alpha^i)$.
As a consequence, the cyclic block code
$\im G_0=\p^{-1}(\ideal{e_1,e_2})=\p^{-1}(\ideal{f})\subseteq\F^n$ has designed distance
$n-1$.
The second block code $\ideal{e_2,\,e_3}$ is simply the image of the first one under the
map~$\sigma$.
Since~$\sigma$ is weight-preserving this shows that the second code has distance~$n-1$,
too.
But now it is clear that the convolutional code~$\cC$ has distance $2(n-1)$ since for
each message $u=\sum_{j=0}^N u_jz^j\in\F[z]^2$, where $u_0\not=0\not=u_N$, the
corresponding codeword $uG$ has constant term $u_0G_0$ and highest coefficient $u_NG_1$,
both of weight at least~$n-1$.
\\
In the same way one can proceed and consider the unit memory code generated by the
polynomial $g=(e_1+e_2+e_3)(1+z)$.
Again, the matrix $M:=\xi(g)$ shows that $g$ is semi-reduced and basic and thus
$\cC=\p^{-1}(\lideal{g})$ is a $3$-dimensional $\sigma$-CCC with all Forney indices
being~$1$.
In this case, \cite[Prop.~4.1]{GS06p} tells us that, if $n\geq 6$, the Griesmer
bound for these parameters is given by the number $2(n-2)+1$.
In the same way as above one can show that the codes just constructed have distance
$2(n-2)$, that is, they fail the Griesmer bound by~$1$.
Proceeding in the same way for arbitrary $k\leq \frac{n}{2}$, one obtains
$k$-dimensional unit memory $\sigma$-CCC's having distance $2(n-k+1)$ which is $k-2$
below the corresponding Griesmer bound.
\end{exa}

\begin{exa}\label{Exa:reduction2}
Let $q=5$ and $n=4$ and $\bar{g}$ be as in Example~\ref{Exa:reduction}.
Write $\F=\F_5$.
One easily checks that the matrix $\bar{M}=\xi(\bar{g})$ given
in\eqnref{eq:barM} is in~$\cMbasic$.
Thus, the submodule $\cC=\p^{-1}(\lideal{\bar{g}})\subseteq\F[z]^4$ is a
$3$-dimensional $\sigma$-CCC with Forney indices $1,\,1,\,2$.
In order to compute a minimal encoder $G\in\F[z]^{3\times4}$ of~$\cC$ we will
apply Theorem~\ref{T-ideals}(4).
From~$\bar{g}$ as given in Example~\ref{Exa:reduction} we obtain the components
$\bar{g}^{(1)}=2e_1+ze_2,\,\bar{g}^{(2)}=e_2+3ze_3,\,
 \bar{g}^{(3)}=e_3+ze_4+4z^2e_1$.
Identifying $f\in\F[x]/{\ideal{x^4-1}}$ with $[f(1),f(2),f(4),f(3)]\in A$ and
using the map~$\p$ from\eqnref{eq:p} we arrive at the minimal encoder matrix
\[
   G=\begin{pmatrix}\p^{-1}(\bar{g}^{(1)})\\ \p^{-1}(\bar{g}^{(2)})\\
                    \p^{-1}(\bar{g}^{(3)})\end{pmatrix}
    =\begin{pmatrix}4z+3& 2z+3& z+3&3z+3\\2z+4&3z+2& 2z+1& 3z+3\\
                    z^2+4z+4& z^2+3z+1& z^2+z+4&z^2+2z+1\end{pmatrix}.
\]
Using some computer algebra routine one checks that this code attains the
Griesmer bound (see \cite[Eq.~(1.3)]{GL06}).
Precisely, its distance is~$6$, which is the largest distance possible for
any $3$-dimensional code of length~$4$ over~$\F_5$ with the same Forney indices.
\end{exa}

Let us now turn to the existence of $\sigma$-CCC's with prescribed algebraic parameters.
Corollary~\ref{cor:main corollary} raises the question whether for all $1\leq k\leq n-1$
and $\nu_1,\ldots,\nu_k\in\N_0$ there exists a $k$-dimensional $\sigma$-CCC in $\F[z]^n$
with Forney indices $\nu_1,\ldots,\nu_k$.
The rest of this section will be devoted to this problem.
We will start with showing that this problem can be split into two subproblems one of
which is purely combinatorial.

Put
\begin{equation}\label{eq:chessboard}
  \hat{D}:=\begin{pmatrix} 0&1&2&\ldots&n-2&n-1\\ n-1& 0& 1 & \ldots &n-3&n-2\\
                         \vdots&\vdots& \vdots&\ddots&\vdots&\vdots
                       \\ 2&3&4&\ldots&0&1\\ 1&2&3&\ldots& n-1&0\end{pmatrix},
\end{equation}
thus,
\begin{equation}\label{eq:Dhatab}
  \hat{D}_{ab}=\left\{\begin{array}{ll} b-a,&\text{if }b\geq a\\[.5ex]
                      n+b-a,&\text{if }b<a.\end{array}\right\}=b-a\mod n.
\end{equation}
The role of the matrix~$\hat{D}$ is explained by the fact that for every matrix
$M=(m_{ab})\in\cM$ we have
\begin{equation}\label{eq:Dhat}
   \cD(M)_{ab}=\left\{\begin{array}{ll}
                    n\deg m_{ab}+\hat{D}_{ab}&\text{if }b\geq a,\\[.5ex]
                    n(\deg m_{ab}-1)+\hat{D}_{ab}&\text{if }b< a.
               \end{array}\right.
\end{equation}

The combinatorial problem we need to consider shows some close resemblance with the
classical rook problem.
\begin{prob}[Modified Rook Problem]\label{P-chessboard}
Let $r_1,\ldots,r_k$ be any (not necessarily different) numbers in $\{0,\ldots,n-1\}$.
Can we find these numbers in the matrix~$\hat{D}$ such that they appear in
pairwise different rows and pairwise different columns?
In other words, can we find distinct numbers $i_1,\ldots, i_k$ and distinct numbers
$j_1,\ldots,j_k$, all in the set $\{1,\ldots,n\}$, such that
\begin{equation}\label{eq:Dhatij}
     \hat{D}_{i_l,j_l}=r_l\text{ for }l=1,\ldots,k.
\end{equation}
\end{prob}
The following slight reformulation will come handy for our purposes.
\begin{rem}\label{R-chessboard}
If Problem~\ref{P-chessboard} is solvable for $r_1,\ldots,r_k$ then we can find these
numbers even in the first~$n-1$ rows of the matrix~$\hat{D}$.
In other words,\eqnref{eq:Dhatij} is true for some distinct numbers
$i_1,\ldots, i_k\in\{1,\ldots,n-1\}$ and distinct numbers
$j_1,\ldots,j_k\in\{1,\ldots,n\}$.
Indeed, suppose we have a solution to~\ref{P-chessboard}, that is, $\hat{D}_{i_l,j_l}=r_l$
for $l=1,\ldots,k$.
Then we may construct a second solution as follows.
There exists some $\alpha\in\{1,\ldots,n\}$ such that $i_l\not=\alpha$ for all~$l$.
Put $a_l=(i_l-\alpha-1\mod n) +1$ and $b_l=(j_l-\alpha-1\mod n)+1$.
Since the numbers $(i_1\mod n),\ldots,(i_k\mod n)$ are pairwise different the same is
true for $a_1,\ldots,a_k$.
Likewise $b_1,\ldots,b_k$ are pairwise different.
Of course, $a_l,\,b_l\in\{1,\ldots,n\}$ for all $l=1,\ldots,k$.
Moreover, by construction $a_l\not=n$ for all $l=1,\ldots,k$, and upon
using\eqnref{eq:Dhatab} we obtain
$\hat{D}_{a_l,b_l}=b_l-a_l\mod n=j_l-i_l\mod n=\hat{D}_{i_l,j_l}=r_l$ for $l=1,\ldots,k$.
\end{rem}

In the next section we will study Problem~\ref{P-chessboard} in some more detail.
Even though we are not able to provide a proof of the solvability for general numbers
$r_1,\ldots,r_k$ we will consider some special cases where we present a complete proof.
We would like to express our strong belief that the problem can be solved for all given
data $r_1,\ldots,r_{n-1}$.
This has been underscored by a routine check with Maple confirming our conjecture for
all $n\leq 10$.

The second problem we need to consider has an affirmative answer and thus can be stated
as a theorem.
The proof will be presented at the end of this section.

\begin{theo}\label{T-MatrixDegrees}
Let $j_1,\ldots,j_{n-1}\in\{1,\ldots,n\}$ be pairwise different and
$d_1,\ldots,d_{n-1}\in\N_0$ be such that
\begin{equation}\label{eq:jd}
   j_i<i\Longrightarrow d_i>0.
\end{equation}
Then there exists a basic matrix $M=(m_{ij})\in\F[t]^{(n-1)\times n}$ with the following
properties:
\begin{romanlist}
\item $\deg m_{ij}\leq d_i$ for $j<j_i$,
\item $\deg m_{ij}=d_i$ for $j=j_i$,
\item $\deg m_{ij}< d_i$ for $j>j_i$,
\item $m_{ii}(0)=1$ for all $i$,
\item $m_{ij}(0)=0$ for $j<i$.
\end{romanlist}
\end{theo}
Notice that the properties~(i) --~(iii) simply tell us that the $i$-th row degree is
given by~$d_i$ and the rightmost entry with degree $d_i$ is in column~$j_i$.
Moreover, observe that without\eqnref{eq:jd} the requirements~(ii) and~(v) would not be
compatible.
Using Proposition~\ref{prop:DegMat} we see that if we extend~$M$ by a zero row at the
bottom we obtain a semi-reduced matrix in $\cMbasic$.

Now we can show the following.
\begin{theo}\label{T-existenceFI}
Let $1\leq k\leq n-1$ and $\nu_1,\ldots,\nu_k\in\N_0$.
Put $r_l=\nu_l\mod n$ for $l=1,\ldots,k$.
If Problem~\ref{P-chessboard} is solvable for $r_1,\ldots,r_k$ then there exists a
$k$-dimensional $\sigma$-CCC in $\F[z]^n$ with Forney indices $\nu_1,\ldots,\nu_k$.
\\
As a consequence, if Problem~\ref{P-chessboard} is solvable for all
$r_1,\ldots,r_{n-1}\in\{0,\ldots,n-1\}$ then for all $1\leq k\leq n-1$ and all
$\nu_1,\ldots,\nu_k\in\N_0$ there exists a $k$-dimensional $\sigma$-CCC in
$\F[z]^n$ with Forney indices $\nu_1,\ldots,\nu_k$.
\end{theo}
\begin{proof}
Let $k\in\{1,\ldots,n-1\}$ and $\nu_1,\ldots,\nu_k\in\N_0$.
Write $\nu_l=\hat{d}_ln+r_l$, where $\hat{d}_l\in\N_0$ and $0\leq r_l\leq n-1$.
By assumption and Remark~\ref{R-chessboard} there exist distinct
$i_1,\ldots,i_k\in\{1,\ldots,n-1\}$ and distinct $j_1,\ldots,j_k\in\{1,\ldots,n\}$
such that $\hat{D}_{i_l,j_l}=r_l$ for $l=1,\ldots,k$.
Pick $1\leq i_{k+1},\ldots,i_{n-1}\leq n-1$ and $1\leq j_{k+1},\ldots,j_{n-1}\leq n$
such that $i_1,\ldots,i_{n-1}$ as well as $j_1,\ldots, j_{n-1}$ are pairwise different.
Define $r_l:=\hat{D}_{i_l,j_l}$ for $l=k+1,\ldots,n-1$ and put $\hat{d_l}:=0$ and
$\nu_l:=r_l$ for $l=k+1,\ldots,n-1$.
Now we re-index the numbers $r_1,\ldots,r_{n-1}$ in order to obtain
\begin{equation}\label{eq:Dhatl}
   \hat{D}_{l,j_l}=r_l\text{ for }l=1,\ldots,n-1.
\end{equation}
Define
\begin{equation}\label{eq:dl}
  d_l:=\left\{\begin{array}{ll} \hat{d}_l&\text{if }j_l\geq l,\\ \hat{d}_l+1&\text{if } j_l<l.\end{array}\right.
\end{equation}
Then\eqnref{eq:jd} is true and Theorem~\ref{T-MatrixDegrees} guarantees the existence of
a matrix $\tilde{M}=(m_{ij})\in\F[t]^{(n-1)\times n}$ satisfying (i)~--~(v).
Extending~$\tilde{M}$ by a zero row at the bottom results in a semi-reduced matrix
$M\in\cMbasic$, see also Proposition~\ref{prop:DegMat}(4).
Using\eqnref{eq:Dhat}, we see that the maxima in the first $n-1$ rows of~$\cD(M)$ are
given by
\[
   \cD(M)_{l,j_l}=\left\{\begin{array}{ll}
        nd_l+\hat{D}_{l,j_l}=n \hat{d}_l+r_l=\nu_l, &\text{if }j_l\geq l,\\[.5ex]
        n(d_l-1)+\hat{D}_{l,j_l}=n\hat{d}_l+r_l=\nu_l,&\text{if }j_l< l.\end{array}\right.
\]
Finally, deleting the rows of~$M$ corresponding to the $n-1-k$ artificially added
indices $\nu_{k+1},\ldots,\nu_{n-1}$ (in the original ordering) we obtain a
semi-reduced matrix $N\in\cMbasic$, and Corollary~\ref{cor:main corollary} shows that
$\p^{-1}(\lideal{\xi^{-1}(N)})$ is a $k$-dimensional code with Forney indices
$\nu_1,\ldots,\nu_k$.
\end{proof}

\begin{exa}\label{Exa:FIs}
Suppose we want a $3$-dimensional $\sigma$-CCC $\cC\subseteq\F_5[z]^4$ with Forney
indices $4,\,3,\,3$. Thus, $q=5$ and $n=4$.
The remainders modulo~$n$ of the desired Forney indices are $0,\,3,\,3$, and by
inspection we find
$\hat{D}_{1,4}=3=:r_1,\; \hat{D}_{2,1}=3=:r_2,\; \hat{D}_{3,3}=0=:r_3$.
Thus,\eqnref{eq:Dhatl} is true for $(j_1,j_2,j_3)=(4,1,3)$.
This gives us the ordering $\nu_1=3,\,\nu_2=3,\,\nu_3=4$ of the Forney indices, and we
have $\hat{d}_1=\hat{d}_2=0$ and $\hat{d}_3=1$.
Following\eqnref{eq:dl} we put $d_1=0,\,d_2=1,\,d_3=1$.
Then we have all data for Theorem~\ref{T-MatrixDegrees}.
One easily sees that the matrix
\[
   M=\begin{pmatrix}1&0&0&1\\ t&1&0&0\\ 0&0&1+t&1\end{pmatrix}
\]
is basic and satisfies~(i) --~(v) of that theorem.
Adding a zero row at the bottom gives us a semi-reduced matrix $N\in\cMbasic$, and
using Corollary~\ref{cor:main corollary} we see that
$\cC=\p^{-1}\big(\lideal{\xi^{-1}(N)}\big)$ is a $3$-dimensional
$\sigma$-CCC in $\F_5[z]^4$ with Forney indices $3,\,3,\,4$.
An encoder matrix of~$\cC$ can be obtained as follows.
The rows of~$M$ along with the isomorphism~$\xi$ result in the component polynomials
$g^{(1)}=e_1+z^3e_4,\,g^{(2)}=e_2+z^3e_1,\,g^{(3)}=e_3+ze_4+z^4e_3$.
Using the same identification of~$A$ with $\F_5[z]/\ideal{x^4-1}$ as in
Example~\ref{Exa:reduction2} we obtain from Theorem~\ref{T-ideals}(4) that
\[
    G=\begin{pmatrix}\p^{-1}(g^{(1)})\\ \p^{-1}(g^{(2)})\\ \p^{-1}(g^{(3)})\end{pmatrix}
     =\begin{pmatrix}4z^3+4& 3z^3+4&z^3+4&2z^3+4\\
        4z^3+4&4z^3+2&4z^3+1&4z^3+3\\4z^4+4z+4&z^4+3z+1&4z^4+z+4&z^4+2z+1\end{pmatrix}
\]
is a minimal encoder of~$\cC$.
It should be pointed out that the distance of this code is far from being optimal.
This is due to the abundance of zeros in the matrix~$M$ causing many zero coefficients
in the matrix~$G$.
\end{exa}

Let us briefly mention the following converse of Theorem~\ref{T-existenceFI}.
Indeed, it is easy to see that the existence of $(n-1)$-dimensional
$\sigma$-CCC's with arbitrarily prescribed Forney indices implies the solvability of
Problem~\ref{P-chessboard} for any given numbers
$r_1,\ldots,r_{n-1}\in\{0,\ldots,n-1\}$.
In more detail, the existence of such codes implies the existence of semi-reduced
polynomials~$g\in\Azs$ with support satisfying $|T_g|=n-1$ and arbitrarily
given degrees of its nonzero components.
Using Proposition~\ref{prop:DegMat}(3) this shows the solvability of
Problem~\ref{P-chessboard} for $k=n-1$, and thus for arbitrary $k\in\{1,\ldots,n-1\}$.

While the general formulation of Theorem~\ref{T-existenceFI} is based on the assumption
that we can solve Problem~\ref{P-chessboard} we have some specific cases with
fully established existence results.

\begin{theo}\label{T-kCCC}
Let $1\leq k\leq \frac{n+1}{2}$.
Then for all $\nu_1,\ldots,\nu_k\in\N_0$ there exists a $k$-dimensional $\sigma$-CCC
having Forney indices $\nu_1,\ldots,\nu_k$.
\end{theo}
\begin{proof}
Using Theorem~\ref{T-existenceFI} it suffices to show that Problem~\ref{P-chessboard}
can be solved for any given numbers $r_1,\ldots,r_k\in\{0,\ldots,n-1\}$ if
$k\leq\frac{n+1}{2}$.
First of all, it is clear that there exists $j_1$ such that $\hat{D}_{1,j_1}=r_1$, and
we will proceed by induction.
Thus, let us assume that we found distinct indices $i_1,\ldots,i_{k-1}$ and
$j_1,\ldots,j_{k-1}$ such that $\hat{D}_{i_l,j_l}=r_l$ for $l=1,\ldots,k-1$.
After a suitable permutation of the rows and columns of~$\hat{D}$ we obtain a matrix
\[
   \tilde{D}=\begin{pmatrix}\tilde{D}_1&\tilde{D}_2\\\tilde{D}_3&\tilde{D}_4\end{pmatrix}
   \text{ where }\tilde{D}_1=\text{diag}(r_1,\ldots,r_{k-1})\in\Z^{(k-1)\times(k-1)}
\]
and where the other matrices are of fitting sizes.
In particular, the matrix~$\tilde{D}_3$ is of size $(n-k+1)\times(k-1)$.
Since the entries of each row (resp.\ column) of~$\tilde{D}$ are pairwise different and,
by assumption, $n-k+1>k-1$ there exists at least one row of~$\tilde{D}_3$ that does not
contain the entry~$r_k$.
But then~$r_k$ must occur in the submatrix~$\tilde{D}_4$, and therefore we have found
$r_1,\ldots,r_k$ in the matrix~$\hat{D}$ in pairwise different columns and rows.
\end{proof}

One might wonder whether the last result can be extended to codes with arbitrary
dimension by using dual codes.
Recall that the dual~$\cC^{\perp}$ of a $k$-dimensional code~$\cC\subseteq\F[z]^n$ has
dimension $n-k$.
However, it is a well-known fact in convolutional coding theory that the Forney indices
of the dual code are not determined by the Forney indices of the given code.
This is also true in the special case of $\sigma$-CCC's as one can easily see by some
examples.
As a consequence, Theorem~\ref{T-kCCC} does not imply any existence results for codes
with higher dimension.

In the next section we will show that we can solve Problem~\ref{P-chessboard} for
parameters $r_1,\ldots,r_{n-1}$ that attain at most two different values,
see Proposition~\ref{P-twodiff}.
Consequently, we have the following result.
\begin{theo}\label{T-existFI}
Let $1\leq k\leq n-1$ and $\nu_1,\ldots,\nu_k\in\N_0$.
If $|\{\nu_1\mod n,\ldots,\nu_k\mod n\}|\leq 2$ then there exists a $k$-dimensional $\sigma$-CCC in $\F[z]^n$ with Forney indices $\nu_1,\ldots,\nu_k$.
\end{theo}

We will close this section with the

\noindent{\sc Proof of Theorem~\ref{T-MatrixDegrees}:}
\\[.5ex]
We assume that~$\F$ is any finite field and $n\geq 2$.
The following notation will be helpful.
For a matrix $A\in\F[t]^{n\times (n+1)}$ and $l=1,\ldots,n+1$ let~$A^{(l)}$ denote the
$n$-minor of~$A$ obtained by omitting column~$l$.
Recall that~$A$ is basic if and only if the polynomials $A^{(1)},\ldots,A^{(n+1)}$ are
coprime.

We will prove even more than stated in Theorem~\ref{T-MatrixDegrees}.
We will show that for the given data there exists a matrix $M$ satisfying the
requirements of the theorem and with the following additional properties:
\begin{romanlist}\setcounter{abc}{5}
\item The only nonzero elements below the diagonal of~$M$ are at the positions $(i,j_i)$
      where $j_i<i$ and they are of the form~$t^{d_i}$. This, of course, implies~(v).
\item If $j_i<i$, then the only nonconstant element in the $i$-th row is at position
      $(i,j_i)$.
\end{romanlist}
Using~(vi), part~(vii) tells us that if $j_i<i$ then all elements at positions $(i,j)$
where $j\geq i$ are constant.
\\[.6ex]
We will proceed by induction on~$n$.
Let $n=2$ and $j_1,\,d_1$ be given.
For $j_1=1$ the matrices $M=(1+t^{d_1},1)$, if $d_1>0$, and $M=(1,0)$, if $d_1=0$,
are basic and have the properties~(i) --~(vii).
If $j_1=2$, the matrix $M=(1,t^{d_1})$ satisfies all requirements.
\\[.6ex]
Let now $n\geq2$ and assume that for all possible $j_1,\ldots,j_{n-1}$ and
$d_1,\ldots,d_{n-1}\in\N_0$ a basic matrix $M\in\F[t]^{(n-1)\times n}$
satisfying~(i) --~(vii) exists.
Throughout this proof we will call such a matrix a solution for the parameters
$(j_1,\ldots,j_{n-1};d_1,\ldots,d_{n-1})$.
Notice that due to~(iv) and~(v) we have $t\nmid M^{(n)}$.
\\[.6ex]
Assume now we have pairwise different indices $j_1,\ldots,j_{n}\in\{1,\ldots,n+1\}$ and
integers $d_1,\ldots,d_{n}\in\N_0$ such that\eqnref{eq:jd} holds true.
We will show the existence of a basic matrix $M\in\F[t]^{n\times(n+1)}$ satisfying~(i)
--~(vii) separately for each of the following cases.
\\[.6ex]
\underline{Case 1: $j_n=n+1$.} Then $j_1,\ldots,j_{n-1}\leq n$ and by induction
hypothesis there exists a solution
$\hat{M}=(\hat{m}_1,\ldots,\hat{m}_n)\in\F[t]^{(n-1)\times n}$ for the parameters
$(j_1,\ldots,j_{n-1};d_1,\ldots,d_{n-1})$.
Put
\[
    M=\left(\!\begin{array}{cccc|c}\hat{m}_1&\cdots&\hat{m}_{n-1}&\hat{m}_n& 0\\[.2ex]
    \hline
    0     &\cdots&  0          &1        & t^{d_n}\end{array}\!\right)
    \in\F[t]^{n\times(n+1)}.
\]
Since $M_{n,n}=1$ and~$\hat{M}$ satisfies~(i) --~(vii) the same is true for~$M$.
Moreover,~$M$ is basic as we can see by considering the $n$-minors.
Indeed, we have $M^{(i)}=\pm t^{d_n}\hat{M}^{(i)}$ for $i=1,\ldots,n$ and
$M^{(n+1)}=\pm\hat{M}^{(n)}$.
Now the coprimeness of $M^{(1)},\ldots,M^{(n+1)}$ follows from the basicness
of~$\hat{M}$ along with the fact that $t\nmid \hat{M}^{(n)}$.
\\[.6ex]
\underline{Case 2: $j_n=n$ and $d_n=0$.} For $i=1,\ldots,n-1$ put
\[
     \hat{j}_i=\left\{\begin{array}{ll}j_i&\text{if } j_i<n\\ n&\text{if } j_i=n+1
               \end{array}\right.
\]
(notice that the case $j_i=n+1$ need not occur).
Then the indices $(\hat{j}_1,\ldots,\hat{j}_{n-1};d_1,\ldots,d_{n-1})$
satisfy\eqnref{eq:jd} and thus there exists a solution
$\hat{M}=(\hat{m}_1,\ldots,\hat{m}_n)\in\F[t]^{(n-1)\times n}$ for these parameters.
Put
\[
   M=\left(\!\begin{array}{cccc|c}\hat{m}_1&\cdots&\hat{m}_{n-1}&0&\hat{m}_n\\[.2ex]
   \hline
   0     &\cdots&  0          &1        &0\end{array}\!\right)\in\F[t]^{n\times(n+1)}.
\]
It is easy to see that~$M$ is a solution for $(j_1,\ldots,j_n;d_1,\ldots,d_n)$.
\\[.6ex]
\underline{Case 3: $j_n=n$ and $d_n>0$.} Choose~$\hat{M}$ as in the previous case and
put
\[
   M=\left(\!\begin{array}{cccc|c}
      \hat{m}_1&\cdots&\hat{m}_{n-1}&\hat{m}_n&\hat{m}_n\\[.2ex] \hline
      0     &\cdots&  0          &1+t^{d_n}&1\end{array}\!\right)\in\F[t]^{n\times(n+1)}.
\]
In this case the $n$-minors of~$M$ are given by $M^{(n+1)}=\pm(1+t^{d_n})\hat{M}^{(n)}$
and $M^{(n)}=\pm\hat{M}^{(n)}$, whereas for $i=1,\ldots,n-1$ we have by expansion along
the last row
$M^{(i)}=\pm\big((1+t^{d_n})\hat{M}^{(i)}-\hat{M}^{(i)}\big)=\pm t^{d_n}\hat{M}^{(i)}$.
Again, basicness of~$\hat{M}$ along with the fact that $t\nmid \hat{M}^{(n)}$ implies
basicness of~$M$.
It is easy to see that $M$ satisfies the properties~(i) --~(vii).
\\[.6ex]
\underline{Case 4: $j_n=:\alpha<n$.} Then $d_n>0$.
We have to distinguish further cases.
\\[.6ex]
(a) If $j_i\leq n$ for all $i=1,\ldots,n-1$ let $\hat{M}=(\hat{m}_1,\ldots,\hat{m}_n)$
be a solution for the parameters $(j_1,\ldots,j_{n-1};d_1,\ldots,d_{n-1})$ and put
\[
   M=\left(\!\begin{array}{ccccc|c}
   \hat{m}_1&\cdots&\hat{m}_\alpha&\cdots&\hat{m}_n&0\\[.2ex] \hline
     &      &t^{d_n}       &      &1&1  \end{array}\!\right)\in\F[t]^{n\times(n+1)}.
\]
It is easy to see that $M$ is basic and satisfies~(i) --~(vii).
\\[.6ex]
(b) Suppose now there exists some index $\beta$ such that $j_{\beta}=n+1$ and assume
that $\beta\leq\alpha$.
Then we may proceed as follows.
For $i=1,\ldots,n-1$ put
\[
   \hat{j}_i=\left\{\begin{array}{ll}j_i&\text{if } i\not=\beta\\
          \alpha &\text{if } i=\beta\end{array}\right.
\]
Then $\hat{j}_1,\ldots,\hat{j}_{n-1}$ are pairwise different and
$(\hat{j}_1,\ldots,\hat{j}_{n-1}; d_1,\ldots,d_{n-1})$ satisfy\eqnref{eq:jd}.
Thus by induction hypothesis there exists a solution
$\hat{M}=(\hat{m}_1,\ldots,\hat{m}_n)\in\F[t]^{(n-1)\times n}$ for the parameters
$(\hat{j}_1,\ldots,\hat{j}_{n-1};d_1,\ldots,d_{n-1})$. Put
\[
   M=\left(\!\begin{array}{ccccc|c}
     \hat{m}_1&\cdots&\hat{m}_\alpha&\cdots&\hat{m}_n&\hat{m}_\alpha\\[.2ex] \hline
        &      &t^{d_n}       &      &1&0   \end{array}\!\right)\in\F[t]^{n\times(n+1)},
\]
where for $j\in\{1,\ldots,n-1\}\backslash\{\alpha\}$ a zero entry occurs at the position
$(n,j)$.
This matrix does not yet satisfy~(i)~--(vii), and we will take care of it later on.
Let us first show that~$M$ is basic.
Computing the $n$-minors we obtain
$M^{(n+1)}=\pm t^{d_n}\hat{M}^{(\alpha)}\pm\hat{M}^{(n)}$ and
$M^{(\alpha)}=\pm \hat{M}^{(n)}$ whereas for $i\not\in\{\alpha,n+1\}$ we have
$M^{(i)}=\pm t^{d_n}\hat{M}^{(i)}$ (since the cofactor of the entry~$1$ in the last row
is zero).
Hence again $t\nmid \hat{M}^{(n)}$ together with the basicness of~$\hat{M}$ implies the
basicness of~$M$.
Let us now turn to the properties~(i) --~(vii).
First of all it is easy to see that~$M=(m_{ij})$ satisfies (i),~(ii),~(iv), and~(v).
In particular, in the row $\beta$, where $j_{\beta}=n+1$, we have by construction
$\deg m_{\beta,n+1}=d_{\beta}$ as well as $\deg m_{\beta,j}\leq d_{\beta}$ for $j<n+1$.
Furthermore, properties~(vi) and~(vii) are satisfied since they are true for~$\hat{M}$
along with the facts that $\beta\leq\alpha$ and $j_{\beta}=n+1>\beta$.
Thus, let us turn to~(iii).
By construction property~(iii) is satisfied for those indices~$i$ for which $j_i<\alpha$.
The only obstacle occurs when $j_i>\alpha$.
In this case we also have $\deg m_{ij}< d_i$ for $j_i<j<n+1$, but the entry in the last
column does not necessarily satisfy this degree constraint.
Due to property~(iii) for~$\hat{M}$ we have instead $\deg m_{i,n+1}\leq d_i$ if
$j_i\geq\alpha$.
We will now perform elementary column operations in order to meet this final degree
constraint.
These column operations will only change the last column of~$M$ and do not destroy any
of the properties mentioned above.
For $l=\alpha+1,\ldots,n$ we consecutively perform the following steps.
If there exists an index~$i_0$ such that~$j_{i_0}=l$ and $\deg m_{i_0,n+1}=d_{i_0}$ then
we add a suitable constant multiple of the $l$-th column of~$M$ to the last column such
that the resulting entry at position $(i_0,n+1)$ has degree strictly less
than~$d_{i_0}$.
This is possible since $\deg m_{i_0,j_{i_0}}=d_{i_0}$.
Now the resulting matrix satisfies~(iii) for all indices~$i$ such that $j_i\leq l$.
Moreover, all other degree constraints remain valid.
In particular, in the row~$\beta$ for which $j_{\beta}=n+1$ we still have
$\deg m_{\beta,n+1}=d_{\beta}$.
This way we finally obtain a solution for the parameters
$(j_1,\ldots,j_n;d_1,\ldots,d_n)$.
\\[.6ex]
(c) It remains to consider the case where there exists some index~$\beta>\alpha$ such
that $j_{\beta}=n+1$.
Notice that $\beta<n$. For $i=1,\ldots,n-1$ put
\[
   (\hat{j}_i,\hat{d}_i)=\left\{\begin{array}{ll}(j_i,d_i)&\text{if } i\not=\beta\\
                 (\alpha,d_n+d_{\beta}) &\text{if } i=\beta\end{array}\right.
\]
Then $\hat{j}_1,\ldots,\hat{j}_{n-1}$ are pairwise different and
$(\hat{j}_1,\ldots,\hat{j}_{n-1}; \hat{d}_1,\ldots,\hat{d}_{n-1})$
satisfies\eqnref{eq:jd} since $\hat{d}_{\beta}=d_n+d_{\beta}\geq d_n>0$.
Thus, by induction hypothesis there is a solution
$\hat{M}=(\hat{m}_1,\ldots,\hat{m}_n)\in\F[t]^{(n-1)\times n}$ for the parameters
$(\hat{j}_1,\ldots,\hat{j}_{n-1};\hat{d}_1,\ldots,\hat{d}_{n-1})$. Put
\[
   M=\left(\!\begin{array}{ccccc|c}
     \hat{m}_1&\cdots&\hat{m}_\alpha&\cdots&\hat{m}_n&0\\[.2ex] \hline
       &      &t^{d_n}       &      &1&1   \end{array}\!\right)\in\F[t]^{n\times(n+1)},
\]
where for $i\in\{1,\ldots,n-1\}\backslash\{\alpha\}$ a zero entry occurs at the position
$(n,i)$.
The $n$-minors of~$M$ are given by
$M^{(n+1)}=\pm\hat{M}^{(n)}\pm t^{d_n}\hat{M}^{(\alpha)}$ and
$M^{(n)}=\pm\hat{M}^{(n)}$, whereas for $i<n$ we have $M^{(i)}=\pm\hat{M}^{(i)}$.
Thus~$M$ is basic.
Furthermore, properties (i) --~(vii) are satisfied for all rows with index $i\not=\beta$,
and only the $\beta$-th row needs to be adjusted.
Since~$\hat{M}$ satisfies~(iv),~(vi), and~(vii) the $\beta$-th row and the $n$-th row
of~$M$ are given by
\begin{align*}
  &(0,\ldots,0,t^{d_n+d_{\beta}},0,\ldots,0,f_{\beta},f_{\beta+1},\ldots,f_n,0),
     \text{ for some }f_l\in\F,\\[.5ex]
  &(0,\ldots,0,\ t^{d_n}\quad\, ,0,\ldots,0,\;0\ ,\ 0\quad ,\ldots,\;1\ ,1),
\end{align*}
respectively, where the entries $t^{d_n+d_{\beta}}$ and $t^{d_n}$ appear in the
$\alpha$-th position, and the entry~$f_{\beta}$ is in the $\beta$-th position.
Moreover, $f_{\beta}=1$ by~(iv).
Now we see that we may subtract $t^{d_{\beta}}$ times the $n$-th row of~$M$ from the
$\beta$-th row in order to obtain a new matrix~$M'$ where the $\beta$-th row is of the
form
\[
     (0,\quad\ldots\ldots\ldots\quad,0,f_{\beta},f_{\beta+1},\ldots,
     f_n-t^{d_{\beta}},-t^{d_{\beta}}),
\]
where still the entry $f_{\beta}$ is in the $\beta$-th position.
Since $j_{\beta}=n+1$ now properties (i)~--(vii) are satisfied for $i=\beta$, whereas
the other rows did not change.
This finally shows that~$M'$ satisfies all requirements~(i) --~(vii).
\hfill$\Box$

\section{The Modified Rook Problem}\label{S-rook}
\setcounter{equation}{0}
In this short section we will briefly discuss Problem~\ref{P-chessboard} for
$k:=n-1$ given numbers.
First of all, notice that the matrix~$\hat{D}$ in\eqnref{eq:chessboard} is the
addition table of the group~$\Z_n:=\Z/n\Z$ if the elements are ordered suitably.
This has actually been used implicitly in Remark~\ref{R-chessboard}.
The additive group~$\Z_n$ allows us to reformulate the problem.
In order to do so let
\[
  \cP:=\{(x_1,\ldots,x_{n-1})\in\Znn\mid x_1,\ldots,x_{n-1}\text{ are pairwise different}\}
\]
and
\[
   \cS:=\{r\in\Znn\mid \exists\; x,\,y\in\cP:\; r=x+y\}.
\]
As a consequence, Problem~\ref{P-chessboard} is solvable for all
$r\in\Znn$ if and only if $\cS=\Znn$.
Here are some simple properties of the set~$\cS$.
\begin{prop}\label{P-S}\
\begin{romanlist}
\item $\gamma\one\in\cS$ for all $\gamma\in\Z_n$.
\item If $r\in\cS$, then $\tau(r)\in\cS$ for all permutations $\tau$ in the symmetric
      group $S_{n-1}$.
\item If $r\in\cS$, then $\gamma r\in\cS$ for all $\gamma\in\Z_n^\times$.
\item If $r\in\cS$, then $r+\gamma\one\in\cS$ for all $\gamma\in\Z_n$.
\item $\cP\subseteq\cS$.
\end{romanlist}
\end{prop}
\begin{proof}
Properties~(i),~(ii) and~(iii) are obvious, whereas~(iv) follows from the fact
that if $x\in\cP$ then $x+\gamma\one\in\cP$ for all $\gamma\in\Z_n$.
Let us now turn to~(v). If $r\in\cP$, then the entries of~$r$ attain $n-1$ of
the~$n$ different elements in~$\Z_n$.
Using~(ii) we may assume that $r=(0,1,\ldots,\alpha-1,\alpha+1,\ldots,n-1)$ for
some $\alpha\in\Z_n$.
Again by~(ii) we have
$r\in\cS\Longleftrightarrow r':=(\alpha+1.\ldots,n-1,0,1,\ldots,\alpha-1)\in\cS$.
By~(iv) this in turn is equivalent to $s:=r'-\alpha\one=(1,2,\ldots,n-1)\in\cS$.
Hence it suffices to show that $s\in\cS$.
For~$n$ being odd one has $s+x=y$ where
\[
   x=(2,3\ldots,{\textstyle \frac{n-1}{2}},{\textstyle\frac{n+1}{2}},
          {\textstyle\frac{n+3}{2}},\ldots,n-1,0),\
   y=(3,5,\ldots,n-2,0,2,\ldots,n-3,n-1).
\]
Since~$n$ is odd $x,\,y$ are in $\cP$ which shows that $s\in\cS$.
For~$n$ even one has $s+x=y$ where
\[
  \begin{matrix}
     x=\!\!&(\quad \frac{n}{2},\;\  &\!\!\frac{n}{2}+1,&\!\!\ldots,&\!\!n-2,
      &\!\!n-1,            &\!\!1,            &\!\!\ldots,&\!\!\frac{n}{2}-2,
      &\frac{n}{2}-1\quad\ \;),\\[.5ex]
     y=\!\!&(\frac{n}{2}+1,&\!\!\frac{n}{2}+3,&\!\!\ldots,&\!\!\frac{n}{2}+n-3,
      &\!\!\frac{n}{2}+n-1,&\!\!\frac{n}{2}+2,&\!\!\ldots,&\!\!\frac{n}{2}+n-4,
      &\!\!\frac{n}{2}+n-2)
  \end{matrix}
\]
Again, $x,\,y\in\cP$, showing the desired result.
\end{proof}

\begin{prop}\label{P-twodiff}
If $r=(r_1,\ldots,r_{n-1})\in\Znn$ has at most two different entries,
that is, $|\{r_1,\ldots,r_{n-1}\}|\leq2$, then $r\in\cS$.
\end{prop}
\begin{proof}
If $r_1=\ldots=r_{n-1}$, then the assertion is in Proposition~\ref{P-S}(i).
Otherwise, using Proposition~\ref{P-S}(ii) we may assume
$r=(\alpha,\ldots,\alpha,\beta,\ldots,\beta)$ for some $\alpha\not=\beta$.
Using part(iv) of that proposition we may even assume that $\alpha=0$.
Thus, let
\[
     r=(\underbrace{0,\ldots,0}_{n-1-f},\underbrace{\beta,\ldots,\beta}_{f})
     \text{ for some }1\leq f\leq n-2.
\]
In order to prove $r\in\cS$ let~$l$ be the additive order of $\beta$ in~$\Z_n$ and put
\[
   t:=(t_1,\ldots,t_n):=(0,\beta,\ldots,(l-1)\beta,1,1+\beta,\ldots,
   1+(l-1)\beta,\ldots,\beta-1,2\beta-1,\ldots,l\beta-1).
\]
That is, the entries of~$t$ are sorted according to the group
$\langle\beta\rangle$ and its cosets.
Now put
\begin{align*}
   x&=(\ \ t_1,\ldots,\ \ t_{n-f-1},\quad t_{n-f+1}\quad,\ldots,\quad t_n\quad),
   \\
   y&=(-t_1,\ldots,-t_{n-f-1},\beta-t_{n-f+1},\ldots,\beta-t_n).
\end{align*}
Then $r=x+y$ and, obviously, $x\in\cP$.
In order to see that $y\in\cP$ notice first that the first $n-f-1$ entries are
obviously pairwise different, and so are the last~$f$ entries.
Assume now $\beta-t_j=-t_i$ for some $n-f+1\leq j\leq n$ and
$1\leq i\leq n-f-1$.
Then $t_j=\beta+t_i$.
But by construction $\beta+t_i=t_{i+1}$ if $i\not\in l\Z$ and
$\beta+t_i=t_{(m-1)l+1}$ if $i=ml$.
Since $j>n-f>i$ this shows that $\beta+t_i\not=t_j$.
Hence $y\in\cP$ and thus $r\in\cS$.
\end{proof}
Notice that the vector~$t$ above could also be defined according to a different
ordering of the cosets of $\langle\beta\rangle$.
This shows, that there are many ways of writing $r=x+y$ for some $x,\,y\in\cP$.

Unfortunately, we are not aware of any way to generalize the last proof to
vectors $r\in\Znn$ with 3 or more different entries.

\section{Extension to General Automorphisms --- An Example}
\label{S-general}
\setcounter{equation}{0}
So far we have studied $\sigma$-CCC's in $\F[z]^n$ where $n\mid(q-1)$ and
where the automorphism~$\sigma$ induces a cycle of length~$n$ on the primitive
idempotents of~$A$.
In this section we will briefly illustrate how the results can
be utilized for general automorphisms if $n\mid(q-1)$.
For ease of notation let us restrict to the following example.

Let $q=8$ and $n=7$.
Consider the automorphism $\sigma\in\AutF(A)$ defined by
\[
  \sigma(e_1)=e_2,\,\sigma(e_2)=e_3,\,\sigma(e_3)=e_1,\,\sigma(e_4)=e_5,\,
  \sigma(e_5)=e_6,\,\sigma(e_6)=e_7,\,\sigma(e_7)=e_4.
\]
In cycle notation this reads as $(e_1,\,e_2,\,e_3)(e_4,\,e_5,\,e_6,\,e_7)$.
Define $A_1:=\F\times\F\times\F$ and $A_2:=\F\times\F\times\F\times\F$ and denote
the primitive idempotents of $A_1$ (resp.\ $A_2$) simply by $e_1,\,e_2,\,e_3$
(resp.\ $e_4,\,e_5,\,e_6,\,e_7$).
Then it is straightforward to establish the isomorphism
\begin{equation}\label{eq:isodirect}
  \Azs\longmapsto A_1[z;\sigma_1]\times A_2[z;\sigma_2],\quad
  g\longrightarrow \Big(g^{(1)}+g^{(2)}+g^{(3)},\,
                        g^{(4)}+g^{(5)}+g^{(6)}+g^{(7)}\Big),
\end{equation}
where for $i=1,2$ the automorphism~$\sigma_i$ on~$A_i$ is defined by the
cycle $(e_1,\,e_2,\,e_3)$ and $(e_4,\,e_5,\,e_6,\,e_7)$, respectively.
Furthermore, a polynomial in $\Azs$ is (semi-)reduced if and only if each
factor in $A_i[z;\sigma_i]$ is (semi-)reduced.
As a consequence, the investigation of left ideals and direct summands in $\Azs$
amounts to the study of the same type of objects in the rings $A_i[z;\sigma_i]$.
Since for these rings the automorphism induces a cycle of maximal length on the
primitive idempotents this brings us to the situation of the previous sections.
Notice, however, that for the ring $A_1[z;\sigma_1]$ the length~$n_1:=3$ is not
a divisor of $q-1$, and thus, $A_1\not\cong\F[x]/{\ideal{x^3-1}}$.
In Remark~\ref{R-general} we mentioned that one can prove the results of
Section~\ref{S-M} in this case as well with the only exception of
Theorem~\ref{T-ideals}(4) which does not make sense anymore.
Along with the isomorphism\eqnref{eq:isodirect} this is sufficient in order to
construct CCC's for this automorphism as well.
Let us illustrate this idea by an example.

\begin{exa}\label{E-generalCCC}
Let $\alpha\in\F_8$ be the primitive element satisfying the identity
$\alpha^3+\alpha+1=0$.
For $i=1,\,2$ let $\xi_i:\;A_i[z;\sigma_i]\longrightarrow\cM_i$ be the
isomorphism with the according matrix ring as introduced in Proposition~\ref{prop:xi}.
The matrices
\[
   M_1=\begin{pmatrix}0&0&0\\0&1&\alpha^4\\0&0&0\end{pmatrix}\in\cM_1
   \text{ and }
   M_2=\begin{pmatrix}\alpha^6&1&\alpha&0\\0&0&0&0\\0&0&\alpha^3&1\\0&0&0&0
      \end{pmatrix}\in\cM_2
\]
are obviously basic and semi-reduced, and thus so are the polynomials
$g_1=\xi_1^{-1}(M_1)=e_2+z\alpha^4e_3\in A_1[z;\sigma_1]$ and
$g_2=\xi_2^{-1}(M_2)=\alpha^6 e_4+\alpha^3 e_6+z(e_5+e_7)+z^2\alpha e_6
 \in A_2[z;\sigma_2]$.
Using the isomorphism in\eqnref{eq:isodirect}
we obtain the semi-reduced and basic polynomial
$g=g_1+g_2=e_2+\alpha^6 e_4+\alpha^3 e_6+z(\alpha^4 e_3+e_5+e_7)+z^2\alpha e_6\in\Azs$.
Its support is given by $T_g=\{2,4,6\}$.
As a consequence, $\lideal{g}$ is a direct summand of rank~$3$ of the left
$\F[z]$-module $\Azs$.
Now~\cite[Thm.~7.13(b)]{GS04} (which is also valid for semi-reduced polynomials)
tells us that the matrix
\[
  G:=\begin{pmatrix}
      \p^{-1}(g^{(2)})\\
      \p^{-1}(g^{(4)})\\
      \p^{-1}(g^{(6)})\end{pmatrix}
      \in\F[z]^{3\times7}
\]
is a minimal encoder for the $\sigma$-CCC $\cC:=\p^{-1}(\lideal{g})\subseteq\F[z]^7$.
By construction, the code has Forney indices~$1,\,2,\,1$.
Of course, basicness and minimality of the matrix~$G$ can also be checked
directly once the rows of the matrix have been computed using the mapping~$\p$
from\eqnref{eq:p}.
Using a computer algebra routine one finds that the distance of~$\cC$ is~$12$.
In other words, the code attains the Griesmer bound, see~\cite[Eq.~(1.3)]{GL06}.
\end{exa}

\section*{Concluding Remarks}
In this paper we studied a particular class of CCC's.
We showed that the existence of such codes with any given algebraic parameters can be
reduced to solving a certain combinatorial problem.
Under the assumption this problem is solvable for all possible instances this shows
that the class of $\sigma$-CCC's is, in a certain sense, as rich as the class of all
CC's.
We strongly believe that the combinatorial problem is solvable for all instances, but
that has to remain open for future research.
Moreover, the potential of our approach needs to be further exploited with
respect to error-correcting properties.
With the exception of the codes in Example~\ref{Exa:2dimCCC} the
considerations so far do not result in classes of codes having provably large distance.

\bibliographystyle{abbrv}
\bibliography{literatureAK,literatureLZ}
\end{document}